\documentclass[12pt,draftcls,onecolumn]{IEEEtran}

\usepackage{graphicx,color}
\usepackage{amsmath,amsthm,amssymb,amsfonts,booktabs}
\usepackage{bbm,dsfont}

\theoremstyle{plain}
\newtheorem{theorem}{Theorem}
\newtheorem{lemma}{Lemma}

\newtheorem{corollary}{Corollary}

\theoremstyle{definition}
\newtheorem{definition}{Definition}
\newtheorem{standing}{Standing Assumption}
\newtheorem{assumption}{Assumption}
\theoremstyle{remark}
\newtheorem{remark}{Remark}

\usepackage{array,subfigure,hyperref}

\newcommand{\R}{\mathbb{R}}
\newcommand{\Z}{\mathbb{Z}}

\newcommand{\X}{\mathbb{X}}
\newcommand{\U}{\mathbb{U}}

\newcommand{\mc}{\mathcal}

\begin{document}

\title{A scenario approach for \\ non-convex control design}
\author{Sergio Grammatico, Xiaojing Zhang, Kostas Margellos, \\ Paul Goulart and  John Lygeros%
\thanks{The authors are with the Automatic Control Laboratory, ETH Zurich, Switzerland. 
E-mail addresses: \{\texttt{grammatico}, \texttt{xiaozhan}, \texttt{margellos}, \texttt{pgoulart}, \texttt{jlygeros}\}\texttt{@control.ee.ethz.ch}.
}
}
\maketitle

\begin{abstract}
Randomized optimization is an established tool for control design with modulated robustness. While for uncertain convex programs there exist randomized approaches with efficient sampling, this is not the case for non-convex problems. Approaches based on statistical learning theory are applicable to non-convex problems, but they usually are conservative in terms of performance and require high sample complexity to achieve the desired probabilistic guarantees.
In this paper, we derive a novel scenario approach for a wide class of random non-convex programs, with a sample complexity similar to that of uncertain convex programs and with probabilistic guarantees that hold not only for the optimal solution of the scenario program, but for all feasible solutions inside a set of a-priori chosen complexity.
We also address measure-theoretic issues for uncertain convex and non-convex programs.
Among the family of non-convex control-design problems that can be addressed via randomization, we apply our scenario approach to randomized Model Predictive Control for chance-constrained nonlinear control-affine systems.
\end{abstract}

\section{Introduction}
Modern control design often relies on the solution of an optimization problem, for instance in robust control synthesis \cite{apkarian:tuan:00, zhou:doyle:glover:97}, Lyapunov-based optimal control \cite{bertsekas, beard:saridis:wen:97}, and Model Predictive Control (MPC) \cite{garcia:prett:morari:89, mayne:rawlings:rao:scokaert:00}.
In almost all practical control applications, the data describing the plant dynamics are \textit{uncertain}. The classic way of dealing with the uncertainty is the robust, also called min-max or worst-case, approach in which the control design has to satisfy the given specifications for \textit{all} possible realizations of the uncertainty. See \cite{bental:nemirovski:02} for an example of robust quadratic Lyapunov function synthesis for interval-uncertain linear systems. The worst-case approach is often formulated as a \textit{robust optimization} problem. However, even though certain classes of robust \textit{convex} problems are known to be computationally tractable \cite{bertsimas:sim:06}, robust convex programs are in general difficult to solve \cite{bental:nemirovski:98, bental:nemirovski:99}. Moreover, from an engineering perspective, robust solutions generally tend to be conservative in terms of performance.

To reduce the conservativism of robust solutions, stochastic programming \cite{prekopa, shapiro:dentcheva:ruszcynski} offers an alternative methodology. 
Unlike the worst-case approach, the constraints of the problem can be treated in a \textit{probabilistic} sense via chance constraints \cite{charnes:cooper:symonds:58, miller:wagner:65}, allowing for constraint violations with chosen low probability. 
The main issue of Chance Constrained Programs (CCPs) is that, without assumptions on the underlying probability distribution, they are in general intractable because multi-dimensional probability integrals must be computed. 

Among the class of chance constrained programs, uncertain \textit{convex} programs have received particular attention \cite{nemirovski:shapiro:04, nemirovski:shapiro:06}. Unfortunately, even for uncertain convex programs, the feasible set of a chance constraint is in general non-convex, which makes optimization under chance constraints problematic \cite[Section 1, pag. 970]{nemirovski:shapiro:06}.

An established and computationally-tractable approach to approximate chance constrained problems is the scenario approximation \cite{nemirovski:shapiro:04}. A feasible solution of the CCP is found with high confidence by solving an optimization problem, called Scenario Program (SP), subject to a finite number of randomly drawn constraints (scenarios). This scenario approach is particularly effective whenever it is possible to generate samples from the uncertainty, since it does not require any further knowledge on the underlying probability distribution. From a practical point of view, this is generally the case for many control-design problems where historical data and/or predictions are available.

The scenario approach for general uncertain (so called random) convex programs was first introduced in \cite{calafiore:campi:05}, and many control-design applications are outlined in \cite{calafiore:campi:06}. The fundamental contribution in these works is the characterization of the number of scenarios, i.e. the sample complexity, needed to guarantee that, with high confidence, the optimal solution of the SP is a feasible solution to the original CCP.
The sample complexity bound was further refined in \cite{campi:garatti:08} where it was shown to be tight for the class of fully-supported problems, and in \cite{calafiore:10, schildbach:fagiano:morari:13} where the concept of Helly's dimension \cite[Definition 3.1]{calafiore:10} and support dimension are introduced, respectively, to reduce the conservativism for non-fully-supported problems.
In \cite{calafiore:10, campi:garatti:11}, the possibility of removing sampled constraints (sampling and discarding) is considered to improve the cost function at the price of decreased feasibility.
It is indeed shown that if the constraints of the SP are removed optimally, then the solution of the SP approaches the actual optimal solution of the original CCP.

While feasibility, optimality and sample complexity of random convex programs are well characterized, to the best of the authors' knowledge, scenario approaches for random \textit{non-convex} programs are less developed. One family of methods comes from statistical learning theory, based on the Vapnik-–Chervonenkis (VC) theory \cite{vapnik:charvonenkis:71, anthony:biggs, vidyasagar:97}, and it is applicable to many non-convex control-design problems \cite{alamo:tempo:camacho:09, calafiore:dabbene:tempo:11, tempo:calafiore:dabbene}.
Contrary to scenario approximations of uncertain convex program \cite{campi:garatti:08, calafiore:10}, the aforementioned methods provide probabilistic guarantees for \textit{all feasible solutions} of the sampled program and not just for \textit{the optimal solution}. This feature is fundamental because the global optimizer of non-convex programs is not numerically computable in general.
Moreover, having probabilistic guarantees for all feasible solutions is of interest in many applications, for instance in \cite{calafiore:campi:elghaoui:02}.
However, the more general probabilistic guarantees of VC theory come at the price of an increased number of required random samples \cite{calafiore:campi:06}. More fundamentally, they depend on the so-called VC-dimension which is in general difficult to compute, or even infinite, in which case VC theory is not applicable \cite{calafiore:campi:05}. 

The aim of this paper is to propose a scenario approach for a wide class of random \textit{non-convex} programs, with moderate sample complexity, providing probabilistic guarantees for \textit{all feasible solutions} in a set of a-priori chosen complexity. In the spirit of \cite{calafiore:campi:05, calafiore:campi:06, campi:garatti:08, calafiore:10}, our results are only based on the ``decision complexity'', while \textit{no assumption} is made on the underlying probability structure. The main contributions of this paper with respect to the existing literature are summarized next.
\begin{itemize}
\item We formulate a scenario approach for the class of random non-convex programs with (possibly) non-convex cost, deterministic (possibly) non-convex constraints, and chance constraints containing functions with separable non-convexity.
For this class of programs, we show via a counterexample that the standard scenario approach is not directly applicable, because Helly's dimension (associated with the global optimum) can be unbounded.
This motivates the development of our technique.

\item We provide a sample complexity bound similar to the one of random convex programs for all feasible solutions in a set of a-priori chosen degree of complexity.

\item We apply our scenario approach methodology to random non-convex programs in the presence of mixed-integer decision variables, with graceful degradation of the associated sample complexity.

\item We apply our scenario approach to randomized Model Predictive Control for nonlinear control-affine systems with chance constraints.

\item We address the measure-theoretic issues regarding the measurability of the optimal value and optimal solutions of random (convex and non-convex) programs, including the well-definiteness of the probability integrals, under minimal mesurability assumptions.
\end{itemize}


The paper is structured as follows. Section \ref{sec:technical-background} presents the technical background and the problem statement. Section \ref{sec:main} presents the main results. Discussions and comparisons with existing methodologies are given in Section \ref{sec:discussion}. Section \ref{sec:control-applications} presents a scenario approach for randomized MPC of nonlinear control affine systems. We conclude the paper in Section \ref{sec:conclusion}. For ease of reading, the Appendices contain: the analytical example with unbounded Helly's dimension (Appendix \ref{app:example}), the technical proofs (Appendix \ref{app:proofs}), and the measure-theoretic results (Appendix \ref{app:measurability}).

\subsection*{Notation}
$\mathbb{R}$ and $\mathbb{Z}$ denote, respectively, the set of real and integer numbers. The notation $\mathbb{Z}[a,b]$ denotes the integer interval $\{a, a+1, ..., b \} \subseteq \Z$. The notation $\textup{conv}(\cdot)$ denotes the convex hull.

\section{Technical background and problem statement} \label{sec:technical-background}
We consider a Chance Constrained Program (CCP) with cost function $J: \R^n \rightarrow \R$, constraint function $g: \R^n \times \R^m \rightarrow \R$, constraint-violation tolerance $\epsilon \in (0,1)$, and admissible set $\mc{X} \subset \R^n$:
\begin{equation}\label{eq:CCP}
\textup{CCP}(\epsilon): \ \left\{ 
\begin{array}{l}
\displaystyle \min_{ x \in \mathcal{X} } \ J(x) \\
\textup{sub. to: } \mathbb{P} \left( \{ \delta \in \Delta \mid g(x,\delta) \leq 0 \}\right) \geq 1-\epsilon.
\end{array}
\right.
\end{equation}
In \eqref{eq:CCP}, $x \in \mc{X}$ is the decision variable and $\delta$ is a random variable defined on a probability space $\left(\Delta, \mc{F}, \mathbb{P}\right)$, with $\Delta \subseteq \R^m$. 
All the measure-theoretic arguments associated with the probability measure $\mathbb{P}$ are addressed in Appendix \ref{app:measurability}.

Throughout the paper, we make the following assumption, partially adopted from \cite[Assumption 1]{calafiore:10}.
\begin{standing}[Regularity] \label{ass:standing}
The set $\mc{X} \subset \R^n$ is compact and convex. For all $\delta \in \Delta \subseteq \R^m$, the mapping $x \mapsto g(x,\delta)$ is convex and lower semicontinuous. For all $x \in \R^n$, the mapping $\delta \mapsto g(x,\delta)$ is measurable.
The function $J$ is lower semicontinuous.
\qed
\end{standing}
The compactness assumption on $\mathcal{X}$, typical of any problem of practical interest, avoids technical difficulties by guaranteeing that any feasible problem instance attains a solution \cite[Section 3.1, pag. 3433]{calafiore:10}. The set $\mc{X} $ is assumed convex without loss of generality\footnote{If $\mathcal{X}$ is not convex, let $\mc{X}' \supset \mc{X}$ be a compact convex superset of $\mc{X}$. Then we can define the indicator function $\chi: \R^n \rightarrow \{ 0, \infty \}$, see \cite[Section 1.A, pag. 6--7]{rockafellar:wets} as $ \chi(x) := 0$ if $x \in \mc{X}$, $\infty$ otherwise.
Then we define the new cost function $J + \chi$, which is lower semicontinuous as well, and finally consider the \textsc{CCP} with convex feasibility set $\mc{X}'$ as
$\min_{ x \in \mc{X}' } \ J(x) + \chi(x) \ \textup{sub. to: } \mathbb{P} \left( \{ \delta \in \Delta \mid g(x,\delta) \leq 0 \} \right) \geq 1-\epsilon$.}.
Measurability of $g(x,\cdot)$ and lower semicontinuity of $J$ are needed to avoid potential measure-theoretic issues, see Appendix \ref{app:measurability} for technical details.

\begin{remark}\label{rem:general-CCP}
The \textsc{CCP} formulation in \eqref{eq:CCP} implicitly includes the more general \textsc{CCP}
\begin{equation}\label{eq:CCP-general}
{\textup{CCP}}'(\epsilon): \left\{ 
\begin{array}{l}
\displaystyle \min_{ x \in \mathcal{X} } \ J(x) \\
\begin{array}{ll}
\textup{sub. to:} & \mathbb{P} \left( \left\{ \delta \in \Delta \mid g(x,\delta) + f(x) \varphi(\delta) \leq 0 \right\} \right) \geq 1-\epsilon\\
                          & h(x) \leq 0,
\end{array}
\end{array}
\right.
\end{equation}
for possibly non-convex functions $f, h: \R^n \rightarrow \R$, $\varphi: \R^m \rightarrow \R$, at the only price of introducing one extra variable\footnote{We follow the lines of \cite[Section 1.A, pag. 6--7]{rockafellar:wets}. The probabilistic constraint becomes $\mathbb{P}\left( \left\{ \delta \in \Delta \mid g(x,\delta) + y \varphi(\delta) \leq 0 \right\}\right) \geq 1 - \epsilon$, while the deterministic constraint becomes $\max\{ h(x), |y - f(x)|\} \leq 0$. We now define the indicator function $\chi: \mc{X} \times \R \rightarrow \{ 0, \infty \}$ as $\chi(x,y) := 0$ if $\max\{ h(x), | y - f(x) | \} \leq 0$, $\infty$ otherwise, in order to get a \textsc{CCP} formulation as in \eqref{eq:CCP}. \\
More generally, we allow for ``separable'' probabilistic constraint of the kind
$ \textstyle \mathbb{P}\left( \left\{ \delta \in \Delta \mid \phi \left( g(x,\delta) + \sum_{i} f_i(x) \varphi_i(\delta) \right) \leq 0  \right\} \right) \geq 1-\epsilon,$
for convex $\phi: \R^{ p \times q } \rightarrow \R$, and possibly non-convex functions $f_i: \R^n \rightarrow \R$.} $y = f(x) \in \mc{Y} := f(\mc{X})$.
{\hfill $\square$}
\end{remark}

It is important to notice that unlike the standard setting of random convex programs \cite{calafiore:campi:05}, we allow the cost function $J$ to be non-convex. Since our results presented later on provide probabilistic guarantees for an \textit{entire} set, rather than for a \textit{single} point, we next give the set-based counterpart of \cite[Definitions 1, 2]{calafiore:campi:06}.
\begin{definition}[Probability of Violation and Feasibility of a Set] \label{def:violation-probability}
The \textit{probability of violation} of a set $\X \subseteq \mathcal{X}$ is defined as 
\begin{equation}\label{eq:violation-probability}
V( \X ) := \sup_{x \in \X} \ \mathbb{P}\left( \{ \delta \in \Delta \mid g(x,\delta)>0 \}\right).
\end{equation}
For any given $\epsilon \in (0,1)$, a set $\X \subseteq \mc{X}$ is \textit{feasible} for $\textup{CCP}(\epsilon)$ in \eqref{eq:CCP} if $V(\X) \leq \epsilon.$
\qed
\end{definition}

In view of Definition \ref{def:violation-probability}, which accounts for the worst-case violation probability on an entire set,
our developments are partially inspired by the following key statement regarding the violation probability of the convex hull of a given set.
\begin{theorem} \label{th:violation-convex-hull}
For given $\X \subseteq \R^n$ and $\epsilon \in (0,1)$, if $V\left(\X\right) \leq \epsilon$, then $V\left( \textup{conv}\left( \X\right) \right) \leq (n+1) \epsilon$.
\hfill $\square$
\end{theorem}
To the best of our knowledge this basic fact has not been exploited in the literature. An immediate consequence of Theorem \ref{th:violation-convex-hull} is that the feasibility set \\ $\mc{X}_{\epsilon} := \left\{ x \in \mc{X} \mid \mathbb{P}\left( \left\{ \delta \in \Delta \mid g(x,\delta)\leq 0 \right\} \right) \geq 1-\epsilon \right\}$ of $\textup{CCP}(\epsilon)$ in \eqref{eq:CCP} satisfies
$$\mc{X}_{\epsilon} \subseteq \textup{conv}\left( \mc{X}_{\epsilon}\right) \subseteq \mc{X}_{(n+1) \epsilon}.$$

Associated with $\textup{CCP}(\epsilon)$ in \eqref{eq:CCP}, we consider a Scenario Program ({SP}) obtained from $N$ independent and identically distributed (i.i.d.) samples $\bar\delta^{(1)}, \bar\delta^{(2)}, ..., \bar\delta^{(N)}$ drawn according to $\mathbb{P}$ \cite[Definition 3]{calafiore:campi:05}.
For a fixed multi-sample $\bar\omega := \left( \bar\delta^{(1)}, \bar\delta^{(2)}, ..., \bar\delta^{(N)}\right) \in \Delta^N$, we consider the SP
\begin{equation}\label{eq:SP}
\textsc{SP}[\bar\omega]: \ \left\{
\begin{array}{l}
\displaystyle \min_{ x \in \mc{X} } \ J(x) \\
\textup{sub. to: }  g(x,\bar\delta^{(i)}) \leq 0 \ \ \forall i \in \Z{[1,N]}.
\end{array}
\right.
\end{equation}

\subsection{Known results on scenario approximations of chance constraints}

In \cite{campi:garatti:08,calafiore:10}, the case $J(x) := c^\top x$ is considered. Under the assumption that, for every multi-sample, the optimizer is unique \cite[Assumption 1]{campi:garatti:08}, \cite[Assumption 2]{calafiore:10} or a suitable tie-breaking rule is adopted \cite[Section 4.1]{calafiore:campi:05} \cite[Section 2.1]{campi:garatti:08}, the optimizer mapping $x^\star(\cdot): \Delta^N \rightarrow \mc{X}$ of $\textsc{SP}[\cdot]$ is such that
\begin{equation} \label{eq:N-Campi-Calafiore}
\mathbb{P}^N \left(  \left\{ \omega \in \Delta^N \mid   V( \{ x^\star(\omega) \} ) > \epsilon \right\} \right) \leq \Phi( \epsilon, n, N ) := \sum_{j=0}^{n-1} { N \choose j } \epsilon^j ( 1 - \epsilon)^{N-j}.
\end{equation}
The above bound is tight for fully-supported problems \cite[Theorem 1, Equation (7)]{campi:garatti:08}, while for non-fully-supported problems it can be improved by replacing $n$ with the so-called Helly's dimension $\zeta \leq n$ \cite[Theorem 3.3]{calafiore:10}. To satisfy the implicit bound \eqref{eq:N-Campi-Calafiore} with right-hand side equal to $\beta \in (0,1)$, it is sufficient to select a sample size \cite[Corollary 5.1]{calafiore:10}, \cite[Lemma 2]{alamo:tempo:luque:ramirez:13}
\begin{equation} \label{eq:N-Campi-Calafiore-1}
N \geq \frac{ \frac{e}{e-1} }{ \epsilon } \left( \zeta - 1 + \textup{ln}\left( \frac{1}{\beta} \right) \right).
\end{equation}
We emphasize that the inequality \eqref{eq:N-Campi-Calafiore} holds only for the probability of violation of the singleton mapping $x^\star(\cdot)$.

Although the only explicit difference between the SP in \eqref{eq:SP} and convex SPs (i.e. with $J(x) := c^\top x$) is the possibly non-convex cost $J$, we show in Appendix \ref{app:example} that Helly's dimension $\zeta$ for the globally optimal value of SP in \eqref{eq:SP} can in general be unbounded. Therefore even for the apparently simple non-convex SP in \eqref{eq:SP} it is impossible to directly apply the classic scenario approach \cite{calafiore:campi:05, calafiore:campi:06, campi:garatti:08, calafiore:10} based on Helly's theorem \cite{floyd:warmuth:95, levin:69}.

For general non-convex programs, VC theory provides upper bounds for the quantity \\
$\mathbb{P}^N \left(  \left\{ \omega \in \Delta^N \mid   V( \X(\omega) ) > \epsilon \right\} \right),$
where $\X(\omega) \subseteq \mc{X}$ is the entire feasible set of $\textsc{SP}[\omega]$, see the discussions in \cite[Section 3.2]{erdogan:iyengar:06}, \cite[Sections IV, V]{alamo:tempo:camacho:09}.

In particular, \cite[Theorem 8.4.1]{anthony:biggs} shows that it suffices to select a sample size
\begin{equation} \label{eq:N-Anthony-Biggs}
N_{\textup{VC}} \geq \frac{4}{ \epsilon } \left( \xi \ \textup{ln}\left( \frac{12}{\epsilon}\right) + \textup{ln}\left( \frac{2}{\beta} \right) \right),
\end{equation}
to guarantee with confidence $1 - \beta$ that any feasible solution of $\textsc{SP}[\bar\omega]$ has probability of violation no larger than $\epsilon$. In \eqref{eq:N-Anthony-Biggs}, $\xi$ is the so-called VC dimension \cite[Definition 10.2]{tempo:calafiore:dabbene}, which encodes the richness of the family of functions $\left\{ \delta \mapsto g(x,\delta) \mid x \in \mc{X} \right\}$ and may be hard to estimate, or even infinite.

\section{Random Non-Convex Programs: \\ Probabilistic Guarantees for an Entire Set} \label{sec:main}

\subsection{Main results} \label{sec:main-results}
We start with a preliminary intuitive statement. We consider a finite number of mappings $x_1^\star, x_2^\star, ..., x_M^\star : \Delta^N \rightarrow \mc{X}$, each one with probabilistic guarantees, and we upper bound their worst-case probability of violation.
\begin{assumption}\label{ass:many-xis}
For given $\epsilon \in (0,1)$, the mappings $x_1^\star, x_2^\star, ..., x_M^\star : \Delta^N \rightarrow \mc{X}$ are such that, for all $k \in \Z{[1,M]}$, we have
$\displaystyle \mathbb{P}^N \left(  \left\{ \omega \in \Delta^N \mid   V( \{ x_k^\star(\omega) \} ) > \epsilon \right\} \right) \leq \beta_k \in (0,1).$
\qed
\end{assumption}

\begin{lemma}\label{lem:many-xis}
Consider the $\textsc{SP}[\bar\omega]$ in \eqref{eq:SP} with $N \geq n$. If Assumption \ref{ass:many-xis} holds, then \\
$\mathbb{P}^N \left(  \left\{ \omega \in \Delta^N \mid   V\left( \left\{ x_1^\star(\omega), x_2^\star(\omega), ..., x_M^\star(\omega) \right\} \right) > \epsilon \right\} \right) 
\ \leq \ \sum_{k=1}^M \beta_k.$
\qed
\end{lemma}

For instance, each $x_i^\star$ may be the optimizer mapping of a convex SP and hence satisfy \eqref{eq:N-Campi-Calafiore} according to \cite{campi:garatti:08, calafiore:10}. In such a case, we get that with probability no smaller than $1 - M \beta$, the set $\left\{ x_1^\star(\omega), x_2^\star(\omega), ..., x_M^\star(\omega) \right\}$ is feasible with respect to Definition \ref{def:violation-probability}, i.e., \\
$\mathbb{P}^N \left(\left\{ \omega \in \Delta^N \mid V\left(  \left\{ x_1^\star(\omega), x_2^\star(\omega), ..., x_M^\star(\omega) \right\} \right) \leq \epsilon \right\} \right) \geq 1-M \beta$.

We may consider the meaning of Lemma \ref{lem:many-xis} in view of the result in \cite[Section 4.2]{anthony:biggs}, and similarly in \cite[Section 4.2]{alamo:tempo:luque:ramirez:13}, where the decision variable $x$ lives in a set $\mc{X}$ of finite cardinality. The main difference here is that Lemma \ref{lem:many-xis} instead relies on a finite number of {mappings} $x_k^\star(\cdot)$, rather than on a finite number of decisions. Each of these mappings is associated with a given upper bound $\beta_k$ on the probability of violating the chance constraint.


We now proceed to our main idea.
We address the \textsc{CCP}($\epsilon$) in \eqref{eq:CCP} through a family of $M$ distinct convex \textsc{SP}s, each with Helly's dimension bounded by some integer $\zeta \in \Z{[1,n]}$.
We consider $M$ cost vectors $c_1, c_2, ..., c_M \in \R^n$, and for each $k \in \Z{[1,M]}$, we define
the following \textsc{SP}, where $\bar\omega := \left( \bar\delta^{(1)}, \bar\delta^{(2)}, ..., \bar\delta^{(N)} \right)$.
\begin{equation}\label{eq:SPk}
\textsc{SP}_k[\bar\omega]: \ \left\{
\begin{array}{l}
\displaystyle \min_{ x \in \mc{C}_k \cap \mc{X} } \ c_k^\top x \\
\textup{sub. to: } g(x,\bar\delta^{(i)}) \leq 0 \ \ \forall i \in \Z{[1,N]} \\
\end{array}
\right.
\end{equation}
The additional convex constraint $x \in \mc{C}_k \subseteq \R^n$ allows to upper bound Helly's dimension by some $\zeta \in \Z[1,n]$, and its choice is hence discussed later on.
Let $x_k^\star(\bar\omega)$ be the optimizer of $\textsc{SP}_k[\bar\omega]$ and assume that it is unique, or a suitable tie-break rule is considered \cite[Section 4.1]{calafiore:campi:05}. We adopt the convention that $x_k^\star(\bar\omega) := \varnothing$ if $\textsc{SP}_k[\bar\omega]$ is not feasible. We notice that if $\textsc{SP}_k[\bar\omega]$ is feasible then the optimizer $x_k^\star(\bar\omega)$ is always finite due to the compactness assumption on $\mathcal{X}$.

For all $\omega \in \Delta^N$, let us consider the convex-hull set 
\begin{equation}\label{eq:XM}
{\X}_M(\omega) := \textup{conv}\left( \{ x_1^\star(\omega), x_2^\star(\omega), ..., x_M^\star(\omega)  \}\right),
\end{equation}
where, for all $k \in \Z{[1,M]}$, $x_k^\star(\cdot)$ is the optimizer mapping of $\textsc{SP}_k[\cdot]$ in \eqref{eq:SPk}. The construction of $\mathbb{X}_M$ is illustrated in Figure \ref{fig:non-convex-scenario}.

\begin{figure}
\begin{center}
\includegraphics[width = 0.5  \columnwidth]{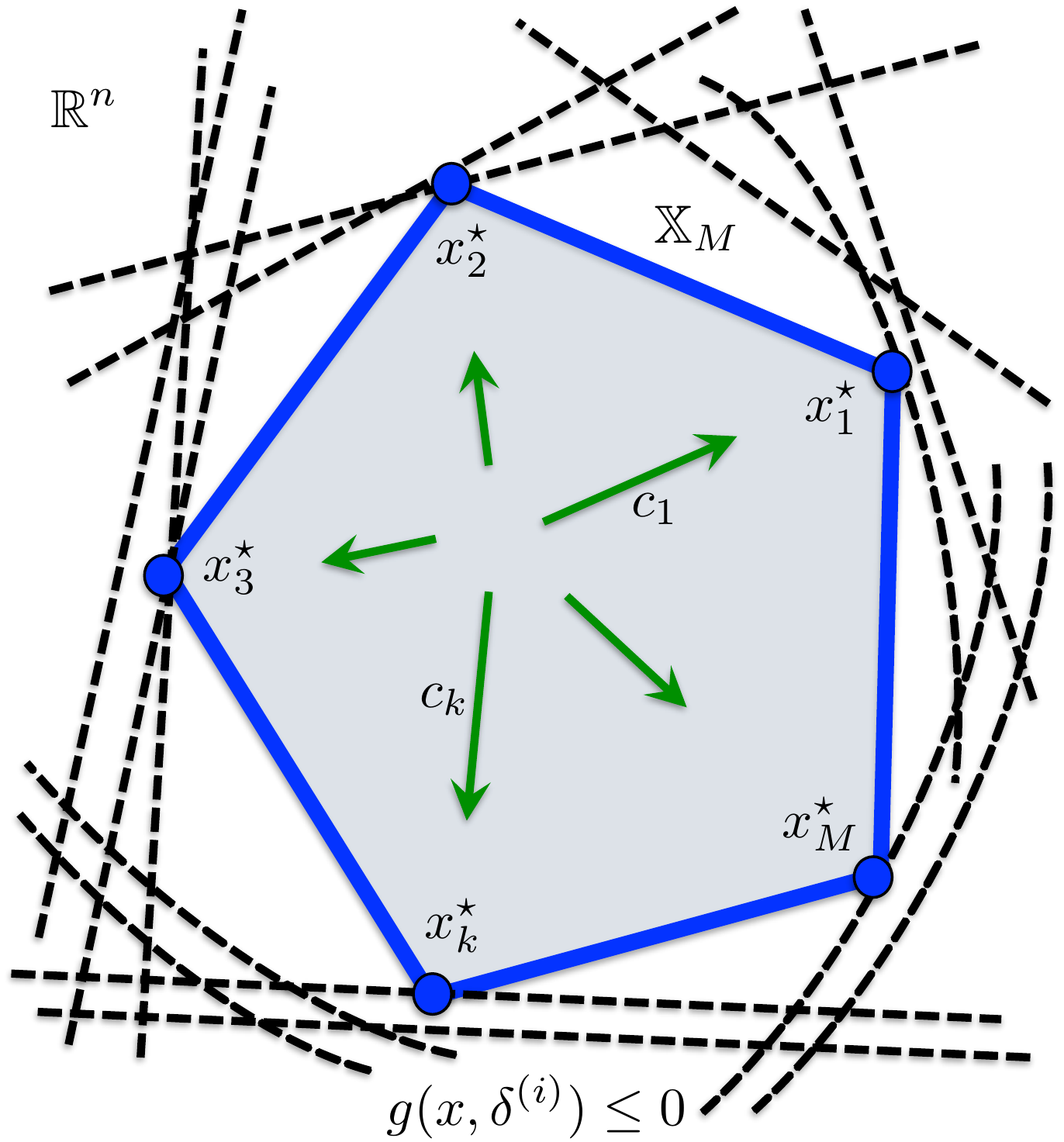}
\end{center}
\caption{The set $\X_M$ is the convex hull of the points $x_1$, $x_2$, ..., $x_M$, where each $x_k$ is the optimizer of $\textup{SP}_k$ in \eqref{eq:SPk} having linear cost $c_k^\top x$.}
\label{fig:non-convex-scenario}
\end{figure}

We are now ready to state our main result.
\begin{theorem}\label{th:bar-Phi}
For each $k \in \Z{[1,M]}$, let $x_k^\star$ and $\zeta$ be, respectively, the optimizer mapping and an upper bound for Helly's dimension of \textsc{SP}$_k$ in \eqref{eq:SPk}, and let ${\X}_M$ be as in \eqref{eq:XM}. Then, for all $\epsilon \in (0,1)$,
\begin{equation} \label{eq:bar-Phi-M}
\mathbb{P}^N\left( \left\{ \omega \in \Delta^N \mid V( \X_M(\omega) ) > \epsilon \right\} \right) \leq M \Phi\left( \frac{\epsilon}{ \min\{n+1, M\} }, \zeta, N \right).
\end{equation}
\qed
\end{theorem}

Following the lines of \cite[Appendix, Proof of Theorem 2]{calafiore:lyons:13}, we can also slightly improve the bound of Theorem \ref{th:bar-Phi} as follows.
\begin{corollary} \label{cor:bar-Phi-best}
For each $k \in \Z{[1,M]}$, let $x_k^\star$ and $\zeta$ be, respectively, the optimizer mapping and Helly's dimension of \textsc{SP}$_k$ in \eqref{eq:SPk}, and let ${\X}_M$ be as in \eqref{eq:XM}. Then, for all $\epsilon \in (0,1)$,
\begin{equation} \label{eq:bar-Phi-M}
\mathbb{P}^N\left( \left\{ \omega \in \Delta^N \mid V( \X_M(\omega) ) > \epsilon \right\} \right) \leq { M \choose n+1 } \Phi\left( \epsilon, \zeta \min\{n+1, M\} , N \right).
\end{equation}
\qed
\end{corollary}

After solving all the $M$ \textsc{SP}s from \eqref{eq:SPk} for the given multi-sample $\bar{\omega} \in \Delta^N$, we can solve the following approximation of \textsc{CCP}($\epsilon$) in \eqref{eq:CCP}: 
\begin{equation}\label{eq:SP-tilde}
\tilde{\textsc{SP}}[\bar\omega]: \ \left\{
\begin{array}{l}
\displaystyle \min_{ x \in \mc{X} } \ J(x) \\
\textup{sub. to: }  x \in \X_M(\bar\omega),
\end{array}
\right.
\end{equation}

and explicitly compute the corresponding sample complexity.
\begin{corollary}\label{cor:bar-Phi}
For each $k \in \Z{[1,M]}$, let $x_k^\star$ and $\zeta$ be, respectively, the optimizer mapping and Helly's dimension of \textsc{SP}$_k$ in \eqref{eq:SPk}, and let ${\X}_M$ be as in \eqref{eq:XM}. For all $\epsilon, \beta \in (0,1)$, if
\begin{equation}\label{eq:our-sample-size}
\displaystyle N \geq \frac{ \frac{e}{e-1}  \min\{n+1,M\}}{\epsilon} \left( \zeta - 1 + \textup{ln}\left( \frac{M}{\beta} \right) \right),
\end{equation}
then $\mathbb{P}^N \left(  \left\{ \omega \in \Delta^N \mid   V( \X_M(\omega) ) \leq \epsilon \right\} \right) \geq 1-\beta$, i.e., with probability no smaller than $1-\beta$, 
any feasible solution of $\tilde{\textsc{SP}}[\bar\omega]$ in \eqref{eq:SP-tilde} is feasible for $\textup{CCP}(\epsilon)$ in \eqref{eq:CCP}.
\qed
\end{corollary}

\begin{remark}\label{rem:SPks}
The constraint $x \in \mc{C}_k$ in \eqref{eq:SPk} provides a way to upper bound Helly's dimension $\zeta$ of $\textsc{SP}_k[\bar\omega]$. Many choices of $\mc{C}_k$ are possible. For instance, $\mc{C}_k := \R^n$ in general only provides the upper bound $\zeta \leq n$.
The minimum upper bound on Helly's dimension for $\textsc{SP}_k[\bar\omega]$ in \eqref{eq:SPk}, i.e. $\zeta = 1$, is obtained whenever $x$ is constrained to live in a linear subspace of dimension one \cite[Lemma 3.8]{schildbach:fagiano:morari:13}. This happens independently from $g$ if, for some fixed $x^0, c_k \in \R^n$, $k = 1, 2, ..., M$, we define 
\begin{equation}\label{eq:Ck-best}
\mc{C}_k := \left\{ x_k^0 + \lambda c_k \in \R^n \mid \lambda \in \R \right\}.
\end{equation} 
With such a chioce of $\mc{C}_k$, $\textsc{SP}_k$ is equivalent to the program
\begin{equation} \label{eq:SP-1dim} 
\min_{ \lambda \in \R } \{- \lambda \} \ \ \textup{sub. to: } (x^0 + \lambda c_k) \in \mc{X}, \ g( x^0 +  \lambda c_k , \delta^i ) \leq 0 \ \ \forall i \in \Z{[1,N]},
\end{equation}
which has unique optimizer and Helly's dimension $\zeta = 1$, since the decision variable $\lambda$ is $1$-dimensional. In this case, the required sample size (for $M\geq n+1$) from \eqref{eq:our-sample-size} is $\left(\frac{e}{e-1} \right) (n+1)/\epsilon \ \textup{ln}\left( M / \beta \right)$, which is \textit{linear} in the number $n$ of decision variables. In particular, if we a-priori know a feasible point $x^0$ for $\textup{CCP}(\epsilon)$ in \eqref{eq:CCP}, then the solution of \eqref{eq:SP-1dim} generates a point $x^0 + \lambda_k^\star(\bar\omega) c_k \in \X(\bar\omega)$ for each $k \in \Z[1,M]$. This additional knowledge is available in many situations of interest \cite[Section 1.1]{care:garatti:campi:11}, for instance in \cite{campi:calafiore:garatti:09, campi:garatti:prandini:09}.

We note that other convex problems can be used in place of $\textsc{SP}_k$ in \eqref{eq:SPk}. For instance, consider the set $R \mathbb{B} := \{ x \in \R^n \mid \left\| x \right\| \leq R \} \subseteq \R^n$, where $R>0$ is such that $\mc{X} \subset R \mathbb{B}$, 
and $M$ points $z_1, z_2, ..., z_M $ on the boundary of $R \mathbb{B}$.
For each $k \in \Z{[1,M]}$, we can define the following \textsc{SP}.
\begin{equation}\label{eq:SPk-alternative}
\textsc{SP}'_k[\bar\omega]: \ \left\{
\begin{array}{l}
\displaystyle \min_{ x \in \mc{X} } \ \left\| x - z_k  \right\| \\
\textup{sub. to: }  g(x,\bar\delta^{(i)}) \leq 0 \ \ \forall i \in \Z{[1,N]}.
\end{array}
\right.
\end{equation}
More generally, the way of selecting the convex problems $\textsc{SP}_k[\bar\omega]$, and hence their associated optimizers $x_k^\star(\bar\omega)$, for $k = 1, 2, ..., M$, is not an essential feature for our  feasibility results.
{\hfill $\square$}
\end{remark}

In view of the approximating, non-convex, problem $\tilde{\textsc{SP}}[\bar\omega]$ in \eqref{eq:SP-tilde} we are interested in getting a ``large'' $\X_M(\bar\omega)$ in \eqref{eq:XM}. 
The choice in \eqref{eq:SPk} is motivated by the fact that the optimal solution $x_k^\star(\bar\omega)$ of ${\textsc{SP}}_k[\bar\omega]$ belongs to the boundary of the actual (convex) feasibily set $\X(\bar\omega) := \left\{ x \in \mc{X} \mid g\left( x, \bar\delta^{(i)} \right) \leq 0 \ \forall i \in \Z[1,N] \right\}$, therefore so do the extreme points of the convex-hull set $\X_M$ in \eqref{eq:XM}, as shown in Figure \ref{fig:non-convex-scenario}.

We finally emphasize that we obtain the probabilistic guarantees in \eqref{th:bar-Phi} for \textit{any} feasible solution of $\tilde{\textsc{SP}}$ in \eqref{eq:SP-tilde}, not just for the optimal solution. This is of practical importance, because $\tilde{\textsc{SP}}[\bar{\omega}]$ is non-convex and hence it is in general impossible to numerically compute its optimal solution.

\subsection{On mixed-integer random non-convex programs}
The results in Lemma \ref{lem:many-xis} and Theorem \ref{th:bar-Phi} can be further exploited to provide probabilistic guarantees for the following class of mixed-integer \textsc{CCP}s.
\begin{equation}\label{eq:CCP-union}
\textup{CCP}^{\textup{m-i}}(\epsilon): \ \left\{ 
\begin{array}{l}
\displaystyle \min_{ (x,j) \in \mathcal{X} \times \Z{[1,L]} } \ J(x) \\
\displaystyle \textup{sub. to: } \mathbb{P} \left( \{ \delta \in \Delta \mid g_j(x,\delta) \leq 0 \}\right) \geq 1-\epsilon,
\end{array}
\right.
\end{equation}
where the functions $g_1, g_2, ..., g_L : \R^n \times \R^m \rightarrow \R$ satisfy the following assumption.
\begin{assumption} \label{ass:gi-convex}
For all $j \in \Z{[1,L]}$ and $\delta \in \Delta$, the mapping $x \mapsto g_j(x,\delta)$ is convex and lower semicontinuous. 
\qed
\end{assumption}
Notice that unlike \cite{calafiore:lyons:fagiano:12}, \cite{mohajerin-esfahani:sutter:lygeros:13}, we also allow for possibly non-convex objective functions $J$.

We also define the probability of violation (of any set $\X \subseteq \mc{X}$) associated with $\textup{CCP}^{\textup{m-i}}(\epsilon)$ in \eqref{eq:CCP-union} as
\begin{equation} \label{eq:V-min}
\textstyle V^{\textup{m-i}}\left( \X \right) := \sup_{ x \in \X } \ \mathbb{P} \left( \left\{ \delta \in \Delta \mid \min_{j \in \Z{[1,L]}} g_j(x,\delta) > 0 \right\}  \right).
\end{equation}
Note that, for all $j \in \Z[1,L]$, it holds $V^{\textup{m-i}}\left( \X \right) \leq \sup_{ x \in \X } \ \mathbb{P} \left( \left\{ \delta \in \Delta \mid g_j(x,\delta) > 0 \right\}  \right)$.

We can proceed similarly to Section \ref{sec:main-results}.  
For fixed multi-sample $\bar\omega \in \Delta^N$, we consider the $M$ cost vectors $c_{1}, c_{2}, ..., c_{M} \in \R^n$ and the convex sets $\mc{C}_{1}, \mc{C}_{2}, ..., \mc{C}_{M} \subseteq \R^n$, so that, for all $(j,k) \in \Z[1,L] \times \Z[1,M]$ we define
\begin{equation}\label{eq:SPkj}
\textsc{SP}^{\textup{m-i}}_{j,k}[\bar\omega]: \ \left\{
\begin{array}{l}
\displaystyle \min_{ x \in \mc{C}_{k} \cap \mc{X} } \ c_k^\top x \\
\textup{sub. to: }  g_j(x,\bar\delta^{(i)}) \leq 0 \ \ \forall i \in \Z{[1,N]}\\
\end{array}
\right.
\end{equation}
with optimizer $x_{j,k}^\star(\bar\omega)$. Then we can define the set $\X_j(\omega)$ as in \eqref{eq:SPk}--\eqref{eq:XM}, i.e.
\begin{equation} \label{eq:Xj}
\X_j(\omega) := \text{conv}\left( \left\{ x_{j,1}^\star(\omega), x_{j,2}^\star(\omega), ..., x_{j,M}^\star(\omega) \right\} \right).
\end{equation}
If $\zeta \in \Z{[1,n]}$ is an upper bound for Helly's dimension of the convex programs $\textsc{SP}^{\textup{m-i}}_{j,k}[\bar\omega]$, then it follows from Theorem \ref{th:bar-Phi} and \eqref{eq:V-min} that, for all $j \in \Z[1,L]$, we have 
\begin{multline}
\mathbb{P}^N\left( \left\{ \omega \in \Delta^N \mid V^{ \textup{m-i} }\left( \X_j(\omega) \right) > \epsilon \right\}\right) \leq \\
\textstyle \mathbb{P}^N\left( \left\{ \omega \in \Delta \mid \sup_{ x \in \X_j(\omega) } \mathbb{P}\left( \left\{ \delta \in \Delta \mid  g_j( x,\delta ) \right\} \right) > \epsilon  \right\} \right) \leq M \Phi\left( \frac{\epsilon}{ \min\{ n+1, M \} }, \zeta , N \right).
\end{multline}

We can then establish the following upper bound on the probability of violation of the union of convex-hull sets constructed above.
\begin{theorem}\label{th:bar-Phi-1}
Suppose Assumption \ref{ass:gi-convex} holds. For each $(j,k) \in \Z{[1,L]} \times \Z{[1,M]}$, let $x_{j,k}^\star$ and $\zeta$ be, respectively, the optimizer mapping and an upper bound for Helly's dimension of \textsc{SP}$_{j,k}^{\textup{m-i}}$ in \eqref{eq:SPkj}; let $\X_j$ be defined as in \eqref{eq:Xj}. Then
\begin{equation} \label{eq:bar-Phi-ML}
\mathbb{P}^N\left( \left\{ \omega \in \Delta^N \mid V^{\textup{m-i}}\left( \cup_{j=1}^L \X_j(\omega) \right) > \epsilon \right\} \right) \leq L M \Phi\left( \frac{\epsilon}{ \min\{n+1,M\}  }, \zeta, N \right).
\end{equation}
\qed
\end{theorem}
We can now approximate the \textup{CCP}$^{\textup{m-i}}(\epsilon)$ in \eqref{eq:CCP-union} by
\begin{equation}\label{eq:SP-tilde-1}
\tilde{\textsc{SP}}^{\textup{m-i}}[\bar\omega]: \ \left\{
\begin{array}{l}
\displaystyle \min_{ (x,j) \in \mc{X} \times \Z[1,L] } \ J(x) \\
\textup{sub. to: }  x \in \X_j(\bar\omega),
\end{array}
\right.
\end{equation}
and state the following lower bound on the required sample size.
\begin{corollary}\label{cor:bar-Phi-1}
Suppose Assumption \ref{ass:gi-convex} holds. For each $(j,k) \in \Z{[1,L]} \times \Z{[1,M]}$, let $x_{j,k}^\star$ and $\zeta$ be, respectively, the optimizer mapping and an upper bound for Helly's dimension of \textsc{SP}$_{j,k}^{\textup{m-i}}$ in \eqref{eq:SPkj}; let $\X_j$ be defined as in \eqref{eq:Xj}.
If
\begin{equation}
\displaystyle N \geq \frac{ \frac{e}{e-1} \min\{n+1,M\} }{\epsilon} \left( \zeta - 1 + \textup{ln}\left( \frac{LM}{\beta} \right) \right),
\end{equation}
then $\mathbb{P}^N \left(  \left\{ \omega \in \Delta^N \mid   V^{\textup{m-i}}\left( \cup_{j=1}^L \X_j(\omega) \right) > \epsilon \right\} \right) \geq 1-\beta$, i.e., with probability no smaller than $1-\beta$, 
any feasible solution of $\tilde{\textsc{SP}}^{\textup{m-i}}[\bar\omega]$ in \eqref{eq:SP-tilde-1} is feasible for $\textup{CCP}^{\textup{m-i}}(\epsilon)$ in \eqref{eq:CCP-union}.
\qed
\end{corollary}

Let us comment on the sample size $N$ given in Corollary \ref{cor:bar-Phi-1}, relative to SP in \eqref{eq:SP-tilde-1}. The formulation in \eqref{eq:CCP-union} subsumes the ones in \cite{calafiore:lyons:fagiano:12} and \cite{mohajerin-esfahani:sutter:lygeros:13}. 
In \cite[Section 4]{mohajerin-esfahani:sutter:lygeros:13} it is shown that it is possible to derive a sample size $N$ which grows linearly with the dimension $d$ of the integer variable $y \in \left(\Z{[ -l/2, l/2 ]} \right)^d$, so that $L := (l+1)^d$ in \eqref{eq:CCP-union}. 
The addition here is that we can also deal with non-convex objective functions $J(x)$ and non-convex deterministic constraints $h(x) \leq 0$ according to Remark \ref{rem:general-CCP}, still maintaining a sample size with logarithmic dependence on $L$, i.e. linear dependence on $d$.
The result in \cite[Theorem 3]{calafiore:lyons:fagiano:12} presents an exponential dependence of the sample size $N$ with respect to the dimension $d$ of the integer variable $y$, but the result therein is slightly more general because it technically covers for a possibly unbounded domain for $y$.

\section{Discussion and comparisons} \label{sec:discussion}

\subsection{Sampling and discarding}
The problem $\textup{SP}_k$ in \eqref{eq:SPk} is also suitable for a sampling-and-discarding approach \cite{calafiore:10, campi:garatti:11}. 
In particular, the aim is to reduce the optimal value of each (convex) SP$_k$ in \eqref{eq:SPk}, and hence enlarge the set $\mathbb{X}_M$ in \eqref{eq:XM}. Indeed, we can a-priori decide that we will discard $r$ of the $N$ samples of the uncertainty. As discussed in \cite{calafiore:10}, any removal algorithm could be employed for the discarding part. Since the optimal constraint discarding is of combinatorial complexity, \cite[Section 5.1]{calafiore:10} proposes greedy algorithms and an approach based on the Lagrange multipliers associated with the constraint functions.
If $N$ and $r$ are taken such that
\begin{equation} \label{eq:N-SaD}
{\zeta + r - 1 \choose r} \Phi\left(\epsilon, \zeta + r, N \right) = {\zeta + r -1\choose r}  \sum_{ i = 1 }^{ \zeta + r -1 } { N \choose i } \epsilon^i ( 1 - \epsilon )^{N-i}  \leq \beta,
\end{equation}
where $\zeta \in \Z{[1,n]}$ is an upper bound on the Helly's dimension of the problem $\textup{SP}_k$ in \eqref{eq:SPk},
then, for all $k \in \Z{[ 1, M ]}$, we have that the optimizer mapping $x_k^*(\cdot)$ of SP$_k[\cdot]$ in \eqref{eq:SPk} (where only $N-r$ constraints are enforced) is such that $ \mathbb{P}^N\left( \left\{ \omega \in \Delta^N \mid V( \{ x_k^\star(\omega)\} > \epsilon ) \right\} \right) \leq \beta $ \cite[Theorem 4.1]{calafiore:10}, \cite[Theorem 2.1]{campi:garatti:11}. Explicit bounds on the sample and removal couple $(N, r)$ are given in \cite[Section 5]{calafiore:10}, \cite[Section 4.3]{campi:garatti:11}.

It then follows from \eqref{eq:N-SaD} that, with $r$ removals over $N$ samples, the optimizer mappings $x_1^\star$, $x_2^\star$, ..., $x_M^\star$ satisfy Assumption \ref{ass:many-xis} with $\beta_k := {\zeta + r - 1 \choose r} \Phi\left(\epsilon, \zeta + r, N \right)$ for all $k \in \Z{[1,M]}$. Therefore, in view of Lemma \ref{lem:many-xis}, we get that the probabilistic guarantees established in Theorem \ref{th:bar-Phi} become
$$ \mathbb{P}^N\left( \left\{ \omega \in \Delta^N \mid V( \X_M(\omega) ) > \epsilon \right\} \right) \leq M {\zeta + r - 1 \choose r} \Phi\left(\frac{\epsilon}{n+1}, \zeta + r, N \right). $$

Since the above inequality relies on $\mathbb{P}^N$, we emphasize that it is possible to remove \textit{different} sets of $r$ samples from each $\textup{SP}_k$. Namely, for all $k \in \Z{[1,M]}$, let $\mc{I}_k \subseteq \Z{[1,M]}$ be a set of indices with cardinality $|\mc{I}_k| = r$. Thus, we can discard the samples $\{ \bar\delta^{(i)} \mid i \in \mc{I}_k \}$ from $\textup{SP}_k$, possibly with $\mc{I}_k \neq \mc{I}_j$ for $k \neq j$.

\subsection{Comparison with the stastical learning theory approach}
Let us compare our sample size in Corollary \ref{cor:bar-Phi} to the corresponding bounds from statistical learning theory based on the VC dimension. 
First, in terms of constraint violation tolerance $\epsilon$, our sample size in \eqref{eq:our-sample-size} grows as $1/\epsilon$ while the sample size provided via the classic statistical learning theory grows as $1/\epsilon^2 \textup{log}( 1/\epsilon^2 )$ \cite[Sections 4, 5]{vidyasagar:01}, \cite[Chapter 8]{tempo:calafiore:dabbene}. An important refinement over the classic result is possible considering the so-called ``one-sided probability of constrained failure'', see for instance \cite[Chapter 8]{anthony:biggs}, \cite[Chapter 7]{vidyasagar:97}, \cite[Section 3]{erdogan:iyengar:06}, \cite[Sections IV, V]{alamo:tempo:camacho:09}. The typical sample size provided in those references is
$ 4/\epsilon \left( \xi \textup{ log}_2\left( 12/{\epsilon}  \right)+ \textup{log}_2\left( 2/{\beta}\right) \right)$,
where $\xi$ is the VC dimension associated with the family of constraint functions $\{g(x, \cdot): \Delta \rightarrow \R \mid x \in \mc{X}\}$. Note that the asymptotic dependence on $\epsilon$ drops from $1/\epsilon^2 \textup{log}( 1/\epsilon^2 )$ to $1/\epsilon \textup{ ln}( 1/ \epsilon )$, but still remains higher than the sample size in \eqref{eq:our-sample-size}.
Second, the sampling-and-discarding approach can be used to enlarge the feasibility domain $\X_M(\bar \omega)$ in \eqref{eq:XM}, as the explicit sample size only grows linearly with the number of removals $r$ \cite[Corollary 5.1]{calafiore:10}. 
On the other hand \cite[Chapter 8, pag. 103]{anthony:biggs},  statistical learning theory approaches cover the possibility of discarding a certain fraction $\rho \in [0,1)$ of the samples, resulting in a sample size of the order of $(\rho + \epsilon)/\epsilon^2 \ \textup{ln}\left( (\rho + \epsilon)/\epsilon^2 \right) $.
Let us indeed denote by $\X_M^{r}(\bar\omega)$ the feasibility set of the methodology of Section \ref{sec:main-results}, defined as in \eqref{eq:XM}, but where each vertex $x_k^\star(\bar\omega)$ is computed 
considering only $N-r$ samples.
It then follows that without any discarding, i.e. for $r=0$, the set $\X_M^{0}(\bar\omega) := \X_M(\bar\omega)$ in \eqref{eq:XM} is always a subset of the entire feasibility set $\X(\bar \omega)$ for any given multi-sample $\bar \omega \in \Delta^N$. However, for $r>0$, the inclusion $\X_M^{r}(\cdot) \subseteq \X(\cdot)$ is no more true therefore the feasibility set constructed in Section \ref{sec:main-results}, together with a sampling-and-discarding approach, is \textit{not necessarily} a subset of the classic statistical learning theory counterpart.
Third, and most important, the sample size in \eqref{eq:our-sample-size} depends only on the dimension $n$ of the decision variable, not on the VC dimension $\xi$ of the constraint function $g$ and, as already mentioned, $\xi$ may be difficult to estimate, or even infinite, in which case VC theory is not applicable.

On the other hand, approaches based on statistical learning theory offer some advantages over our method. They in fact cover general non-convex problems and, without any sampling and discarding, provide probabilistic guarantees for all feasible points, not only for those in a certain subset of given complexity.

\subsection{Comparison with mixed random-robust approach}
An alternative approach based on a mixture of randomized and robust optimization was presented in \cite{margellos:goulart:lygeros:13}. It requires solving a robust problem with the uncertainty being confined in an  appropriately parametrized set, generated in a randomized way to include $(1-\epsilon)$ of the probability mass of the uncertainty, with high confidence. Following this approach one obtains probabilistic guarantees for any feasible solution of the robust problem. In particular, the size of this subset depends on the parametrization of the uncertainty set, which in turn affects the number of scenarios that must be extracted.
However, in contrast to the current paper, the approach in \cite{margellos:goulart:lygeros:13} has some drawbacks listed as follows. First, it is not guaranteed that the a-priori chosen parametrization generates a feasible robust optimization problem. Second, if such robust program is feasible, it is in general conservative in terms of cost and computationally tractable only for a very limited class of non-convex problems. In particular, some additional structure on the dependence on the uncertainty must be assumed. Finally, the method in \cite{margellos:goulart:lygeros:13} comes with no explicit characterization of the probabilistically-feasible subset in the decision variable domain.

\section{Randomized Model Predictive Control of nonlinear control-affine systems and other control applications} \label{sec:control-applications}

\subsection{Randomized Model Predictive Control of nonlinear control-affine systems} \label{sec:MPC}

In this section we extend the results of \cite{calafiore:fagiano:12, schildbach:fagiano:frei:morari:13} to uncertain \textit{nonlinear control-affine} systems of the form
\begin{equation} \label{eq:nonlinear-system}
x^+ = f( x, v) + g(x, v) u,
\end{equation}
where $x \in \R^n$ is the state variable, $u \in \R^m$ is the control variable, and $v \in \mc{V} \subseteq \R^p$ is the uncertain random input. We assume state and control constraints $x \in \X \subseteq \R^n$, $u \in \U \subseteq \R^m$, where $\X$ and $\U$ are compact convex sets. 
We further assume the availability of i.i.d. samples $\bar v^{(1)}, \bar v^{(2)}, ...$ of the uncertain input, drawn according to a possibly-unknown probability measure $\mathbb{P}$ \cite[Definition 3]{calafiore:campi:05}.

For a horizon length $K$, let $\mathbf{u} := \left( u_0, u_1, ..., u_{K-1}\right)$ and $\mathbf{v} := \left( v_0, v_1, ..., v_{K-1}\right)$ denote a control-input and random-input sequence respectively.
We denote by $\phi(k; x, \mathbf{u}, \mathbf{v})$ the state solution of \eqref{eq:nonlinear-system} at time $k \geq 0$, starting from the initial state $x$, under the control-input sequence $\mathbf{u}$ and the random-input sequence $\mathbf{v}$. Likewise, given a control law $\kappa: \X \rightarrow \U$, we denote by $\phi_{\kappa}(k; x, \mathbf{v})$  the state solution of the system $x^+ = f( x,v ) + g(x,v) \kappa(x) $ at time $k \geq 0$, starting from the initial state $x$, under the random-input sequence $\mathbf{v}$.
The solution $\phi(k;x,\mathbf{u},\mathbf{v})$, as well as $\phi_{\kappa}(k;x,\mathbf{v})$, is a random variable itself\footnote{Random solutions, both $\phi(k;x,\mathbf{u},\mathbf{\cdot})$ and $\phi_{\kappa}(k;x,\mathbf{\cdot})$, exist under the assumption that for all $x \in \R^n$, the mapping $\delta \mapsto f(x,\delta) + g(x,\delta)$ is measurable and that $\kappa$ is measurable, see \cite[Section 5.2]{grammatico:subbaraman:teel:13} and Appendix \ref{app:measurability} for technical details.} because under the dependence on the random-input sequence $\mathbf{v}$.

Let $\ell: \R^n \times \R^m \rightarrow \R_{\geq 0}$ be the stage cost, and $\ell_f: \R^n \rightarrow \R_{\geq 0}$ be the terminal cost. We consider the random finite-horizon cost function
\begin{equation}
J(x, \mathbf{u}, \mathbf{v} ) := \ell_f( \phi\left(K; x, { \mathbf u }, {\mathbf v} \right)  ) + \sum_{ k = 0 }^{K-1} \ell( \phi\left(k; x, { \mathbf u }, {\mathbf v} ), u_{k} \right).
\end{equation}

Following \cite[Section 3.1]{schildbach:fagiano:frei:morari:13}, we formulate the multi-stage Stochastic MPC (SMPC) problem 
\begin{equation} \label{eq:SMPC}
\left\{
\begin{array}{l}
\displaystyle \min_{ \mathbf{u} \in \U^K } \ \mathbb{E}^K \left[ J( x, \mathbf{u}, \mathbf{\cdot} ) \right]\\
\begin{array}{l}
\textup{sub. to: } \mathbb{P}^{k} \left( \left\{ \mathbf{v} \in \mathcal{V}^k \mid \phi\left(k; x, \mathbf{u}, \mathbf{v}  \right) \in \X  \right\} \right) \geq 1-\epsilon \quad \forall k \in \Z{[1,K]}
\end{array}
\end{array}
\right.
\end{equation}
and its randomized (non-convex) counterpart
\begin{equation} \label{eq:RMPC}
\textup{SP}^{ \textup{MPC} }[\mathbf{ \bar  v}^{(1)}, \mathbf{\bar v}^{( 2 )}, ...]: \ \left\{
\begin{array}{l}
\displaystyle \min_{ \mathbf{u} \in \U^K } \ \sum_{i \in \mathcal{I}_0}  J( x, \mathbf{ u}, \mathbf{ \bar v}^{(i)} ) \\
\begin{array}{lll}
\textup{sub. to: }  & \phi( 1; x, { \mathbf{u} }, \mathbf{ \bar  v}^{(i)} ) \in \X  & \forall i \in \mathcal{I}_1 \\
 &  \phi( k; x, { \mathbf{u} }, \mathbf{ \bar  v}^{(i)} ) \in \X &  \forall i \in \mathcal{I}_2, \ \forall k \in \Z{[2, K]}, \ 
\end{array}
\end{array}
\right.
\end{equation}
for some disjoint index sets $\mathcal{I}_0, \mathcal{I}_1, \mathcal{I}_2 \subset \Z[1,\infty)$.
The receding horizon control policy is defined as follows. For each time step, we measure the state $x$ and let $\mathbf{u}^\star(x) := \left( u_0^\star, \ldots, u_{K-1}^\star\right)(x)$ be the solution of $\textup{SP}^{ \textup{MPC} }$ in \eqref{eq:RMPC}, for some drawn samples $\{ \mathbf{ \bar  v}^{(1)}, \mathbf{\bar v}^{( 2 )}, ... \}$. The control input $u$ is set to the first element of the computed sequence, namely $u = \kappa(x) := u_0^\star$, which implicitly also depends on the samples extracted to build the optimization program itself.

We next focus on a suitable choice for the sample size, so that the average fraction of \textit{closed-loop} constraint violations ``$x_1 \notin \X, x_2 \notin \X, \ldots, x_t \notin \X$'' is below the desired level $\epsilon$.
It follows from \cite[Section 3]{schildbach:fagiano:frei:morari:13} that this property is actually independent from the cardinalities of $\mathcal{I}_0$ and $\mathcal{I}_2$, i.e. on the number of samples used for the cost function and for the later stages. In fact, under proper assumptions introduced later on, the closed-loop behavior in terms of constraint violations is only influenced by the first-stage constraint, namely by the number $N$ of samples indexed in $\mathcal{I}_1$ \cite[Section 3]{schildbach:fagiano:frei:morari:13}. Without loss of generality, let $\mathcal{I}_1 := \Z[1,N]$ for ease of notation. We refer to \cite{zhang:grammatico:margellos:goulart:lygeros:14ifac} for a discussion on the role of $\mathcal{I}_0$ and $\mathcal{I}_2$ in terms of closed-loop performance. 

In particular, later on we show that our main results of Section \ref{sec:main-results} are directly applicable because the sampled nonlinear MPC program $\textup{SP}^{ \textup{MPC} }$ in \eqref{eq:RMPC} has non-convex cost, due to the nonlinear dynamics in \eqref{eq:nonlinear-system}, and convex first-stage constraint.
Since the program in \eqref{eq:RMPC} is non-convex, and hence the global optimizer is in general not computable exactly,
we adopt the following set-based definition of probability of violation.
\begin{definition}[First-stage probability of violation]
For given $x \in \X$ and $\U_0 \subseteq \U$, the first-stage probability of violation is given by 
\begin{equation*}
V^{ \textup{MPC} }( x, \U_0 ) := \sup_{u \in {\U}_0} \ \mathbb{P}\left( \left\{ v \in \mathcal{V} \mid f(x,v) + g(x,v) u \notin \X \right\}  \right).
\end{equation*}
\qed
\end{definition}

Analogously to Section \ref{sec:main-results}, see Remark \ref{rem:SPks}, we then consider $M$ directions $c_1, c_2, ..., c_M \in \R^m$, and an arbitrary  $\hat{u}_0 \in \R^m$. For instance, but not necessarily,
$\hat{u}_0$ may be a known robustly feasible solution.
For all $j \in \Z{[1,M]}$, we define the following \textsc{SP}, where $\mathbf{\bar v_0} := ( \bar v_0^{(1)}, ..., \bar v_0^{(N)})$.
\begin{equation}\label{eq:SPk-1}
\textup{SP}^{ 1}_j[\mathbf{\bar v_0}]: \ \left\{
\begin{array}{l}
\displaystyle \min_{ \lambda \in \R } \ -\lambda \\
\begin{array}{ll}
\textup{sub. to:}  & f(x, \bar v_0^{(i)}) + g(x, \bar v_0^{(i)}) (\hat{u}_0 + \lambda c_j )  \in \X  \quad \forall i \in \Z{[1,N]} \\
                           &  \hat{u}_0 + \lambda c_j\in \U.
\end{array}
\end{array}
\right.
\end{equation}
Let $\lambda_j^\star$ be the optimizer mapping of $\textup{SP}^{ 1}_j$. If $\textup{SP}^{ 1}_j[\mathbf{\bar v_0}]$ is not feasible, we use the convention that $\lambda_j^\star(\bar\omega) := \varnothing$.
For all the feasible problems $\textup{SP}^{ 1}_j[\mathbf{\bar v_0}]$, we define
\begin{equation} \label{eq:UM}
\U_M( \mathbf{ \bar  v_0} ) := \textup{conv}\left( \{ \hat{u}_0 + \lambda_1(\mathbf{ \bar v_0}) d_1, \hat{u}_0 + \lambda_2(\mathbf{ \bar v_0}) d_2, ..., \hat{u}_0 + \lambda_M(\mathbf{ \bar v_0}) d_M \} \right).
\end{equation}
Finally, we solve the following approximation of $\textup{SP}^{ \textup{MPC} }$ in \eqref{eq:RMPC}.
\begin{equation}\label{eq:RMPC-approx}
\tilde{\textup{SP}}^{ \textup{MPC} }[\mathbf{\bar v}^{(1)}, \mathbf{\bar v}^{( 2 )}, ...]: \ \left\{
\begin{array}{l}
\displaystyle \min_{ \mathbf{u} \in \U^N } \ \sum_{i \in \mathcal{I}_0} J( x, \mathbf{u}, \mathbf{\bar v}^{(i)} ) \\
\begin{array}{lll}
\textup{sub. to:} & u_0 \in \U_M(\mathbf{\bar v_0}) &  \\
 & \phi( k; x, { \mathbf{u} }, \mathbf{ \bar  v}^{(i)} ) \in \X &  \forall i \in \mathcal{I}_2, \ \forall k \in \Z{[2, K]}
\end{array}
\end{array}
\right.
\end{equation}
We can now characterize the required sample complexity for the probability of violation to be, with high confidence, below the desired level.
\begin{theorem} \label{th:N-MPC}
For all $x \in \X$ and $j \in \Z{[1,M]}$, let $\lambda_j^\star$ be the optimizer mapping of $\textup{SP}^1_j$ in \eqref{eq:SPk-1}, let ${\U}_M$ be as in \eqref{eq:UM}, and $\epsilon, \beta \in (0,1)$. Then
\begin{equation}
\mathbb{P}^{N} \left( \left\{ \mathbf{v_0} \in \mathcal{V}^N \mid V^{\textup{MPC}}( x, \U_M( \mathbf{v_0} ) ) > \epsilon  \right\} \right) \leq M \Phi\left( \frac{\epsilon}{ \min\{ m+1, M \} }, 1, N \right).
\end{equation}
Consequently, if
\begin{equation} \label{eq:K-MPC}
\displaystyle N \geq \frac{\frac{e}{e-1} \min\{ m+1, M\}}{\epsilon} \textup{ln}\left( \frac{M}{\beta} \right),
\end{equation}
then $\mathbb{P}^{N} \left( \left\{ \mathbf{v_0} \in \mathcal{V}^N \mid V^{\textup{MPC}}( x, \U_M( \mathbf{v_0} ) ) \leq \epsilon  \right\} \right) \geq 1-\beta$, i.e., with probability no smaller than $1-\beta$, 
any feasible solution of $\tilde{\textup{SP}}^{ \textup{MPC} }$ in \eqref{eq:RMPC-approx} satisfies the first state constraint in \eqref{eq:SMPC}.
\qed
\end{theorem}

The result of Theorem \ref{th:N-MPC} can be exploited to characterize the \textit{expected} closed-loop constraint violation as in \cite[Theorem 14]{schildbach:fagiano:frei:morari:13}, under the following assumption {\cite[Assumption 5]{schildbach:fagiano:frei:morari:13}}.
\begin{assumption}[Recursive feasibility] \label{ass:recursive-feasibility}
$\textup{SP}^{ \textup{MPC} }$ in \eqref{eq:RMPC} admits a feasible solution at every time step almost surely.
\qed
\end{assumption}
\begin{corollary}\label{cor:expected-violations}
Suppose Assumption \ref{ass:recursive-feasibility} holds. For all $x \in \X$ and $ \mathbf{ v} := \left(\mathbf{ v}^{(1)}, ..., \mathbf{ v}^{(N)} \right) \in \mathcal{V}^{KN} $, \\
let $\mathbf{u}(x) := \left( u_0(x), ..., u_{K-1}(x)\right)$ be any feasible solution of $\tilde{\textup{SP}}^{ \textup{MPC} }[\mathbf{ v}]$ in \eqref{eq:RMPC-approx}, and define $\kappa(x) := u_0(x)$. Let $\U_M(k;\mathbf{v} )$ be the set $\U_M( \mathbf{v_0} )$ in \eqref{eq:UM} with $\phi_{\kappa}( k; x, \mathbf{v} )$ in place of $x$.
If $N$ satisfies 
\begin{equation} \label{eq:admissible-K}
\int_0^{1} M \Phi\left( \frac{\nu}{ \min\{ m+1, M \} }, 1, N \right) d \nu 	\ \leq \ \epsilon,
\end{equation}
then\footnote{In \cite[Definition 12]{schildbach:fagiano:frei:morari:13}, a sample size $N$ is called \textit{admissible} if it satisfies $ \int_0^{1} \Phi( \nu, m, N ) d \nu \leq 	\epsilon $, which is the counterpart of \eqref{eq:admissible-K} for random convex programs. For given $\epsilon \in (0,1)$, $m, M>0$, an admissible $K$ satisfying \eqref{eq:admissible-K} can be evaluated via a numerical one-dimensional integration.}, for all $k \geq 0$ it holds that
\begin{equation*}
\displaystyle \mathbb{E}\left[ V^{ \textup{MPC} }\left( \phi_{\kappa}(k; x, \cdot), \U_M(k;\mathbf{\cdot}) \right)  \right]  :=  \int_{ \mathcal{V}^{(KN+1)k} } V^{ \textup{MPC} }\left( \phi_{\kappa}( k; x, \mathbf{v} ), \U_M(k; \mathbf{v} ) \right) \mathbb{P}^{ (KN+1)k }( \mathbf{dv} ) \leq \epsilon. 
\end{equation*}
{\hfill $\square$}
\end{corollary}
The meaning of Corollary \ref{cor:expected-violations} is that the expected closed-loop constraint violation, which can be also interpreted as time-average closed-loop constraint violation \cite[Section 2.1]{zhang:grammatico:margellos:goulart:lygeros:14ifac}, is upper bounded by the specified tolerance $\epsilon$ whenever the sample size $N$ satisfies \eqref{eq:admissible-K}. A similar result was recently shown in \cite[Section 4.2]{schildbach:fagiano:frei:morari:13} for uncertain linear systems and hence here extended to the class of uncertain nonlinear control-affine systems in \eqref{eq:nonlinear-system}.

Numerical simulations of the proposed stochastic nonlinear MPC approach are provided in \cite{zhang:grammatico:margellos:goulart:lygeros:14ifac} for a nonholonomic control-affine system, and the benefits with respect to stochastic linear MPC are shown therein.

\subsection{Other non-convex control-design problems}
Our scenario approach is suitable for many non-convex control-design problems, such as robust analysis and control synthesis \cite{alamo:tempo:camacho:09, alamo:tempo:luque:ramirez:13}. In particular,  in \cite{grammatico:zhang:margellos:goulart:lygeros:14acc} we address control-design via uncertain Bilinear Matrix Inequalities (BMIs) making comparison with the sample complexity based on statistical learning theory, recently derived in \cite{chamanbaz:13}.
Many practical control problems also rely on an uncertain non-convex optimization, for instance reserve scheduling of systems with high wind power generation \cite{vrakopoulou:margellos:lygeros:andersson:13}, aerospace control \cite{wang:stengel:05}, truss structures \cite{calafiore:dabbene:08}.
Other non-convex control problems that can be addressed via randomization arise in the control of switched systems \cite{ishii:basar:tempo:05}, network control \cite{alpcan:basar:tempo:05}, fault detection and isolation \cite{ma:sznaier:lagoa:07}.

\section{Conclusion} \label{sec:conclusion}
We have considered a scenario approach for the class of random non-convex programs with (possibly) non-convex cost, deterministic (possibly) non-convex constraints, and chance constraint containing functions with separable non-convexity. For this class of programs, Helly's dimension can be unbounded. We have derived probabilistic guarantees for all feasible solutions inside a convex set with a-priori chosen complexity, which affects the sample size merely logarithmically. 

Our scenario approach also extends to the case with mixed-integer decision variables.
We have applied our scenario approach to randomized Model Predictive Control for nonlinear control-affine systems with chance constraints, and outlined many non-convex control-design problems as potential applications.

Finally, we have addressed the measure-theoretic issues regarding the measurability of the optimal value and optimal solutions of random (convex and non-convex) programs, including the well-definiteness of the probability integrals. 

\section*{Acknoledgements}
The authors would like to thank Marco Campi, Simone Garatti, Georg Schildbach for fruitful discussions on related topics. 
Research partially supported by Swiss Nano-Tera.ch under the project HeatReserves.

\appendices

\section{Counterexample with unbounded number of support constraints} \label{app:example}
We present an SP, derived from a CCP of the form \eqref{eq:CCP-general}, in which Helly's dimension \cite[Definition 3.1]{calafiore:10} cannot be bounded. Namely, the number of constraints (``support constraints'' \cite[Definition 2.1]{calafiore:10}) needed to characterize the global optimal value equals the number $N$ of samples.
\begin{equation} \label{eq:counterexample}
\textup{SP}_{\textup{ex}}[\bar\omega]: \ \left\{
\begin{array}{l}
\displaystyle \min_{ (x,y,z) \in \R^3 } \ z \\
\begin{array}{ll}
\textup{sub. to: } & z \geq -\sqrt{x^2 + y^2} \\
 & z \geq \cos( \bar\delta^{(i)} ) x + \sin( \bar\delta^{(i)} ) y -1 \quad \forall i \in \Z{[ 1,N ]}.
\end{array}
\end{array}
\right.
\end{equation}
The problem can be also written in the form \eqref{eq:SP}, with non-convex cost $J(x,y) := -\sqrt{x^2 + y^2}$ and non-convex constraints $-\sqrt{x^2 + y^2} \geq \cos( \bar\delta^{(i)} ) x + \sin( \bar\delta^{(i)} ) y -1$. We use the form in \eqref{eq:counterexample} to visualize the optimizing direction $-z$, as shown in Figure \ref{fig:SP1}.

Let the drawn samples be $\bar\delta^{(i)} = (i-1)\frac{2 \pi}{N}$, for $i = 1, 2, ..., N$. Namely, we divide the $2 \pi$-angle into $N$ parts, so that $\bar\delta^{(1)} = 0$ and $\bar\delta^{(i+1)} = \bar\delta^{(i)} + \frac{2 \pi}{N}$ for all $i \in \Z{[1,N-1]}$. We take $N \geq 5$ as $\frac{2 \pi}{N} \in (0, \pi/2)$ simplifies the analysis.

\begin{figure} 
\begin{center}
\includegraphics[width = 0.8 \columnwidth]{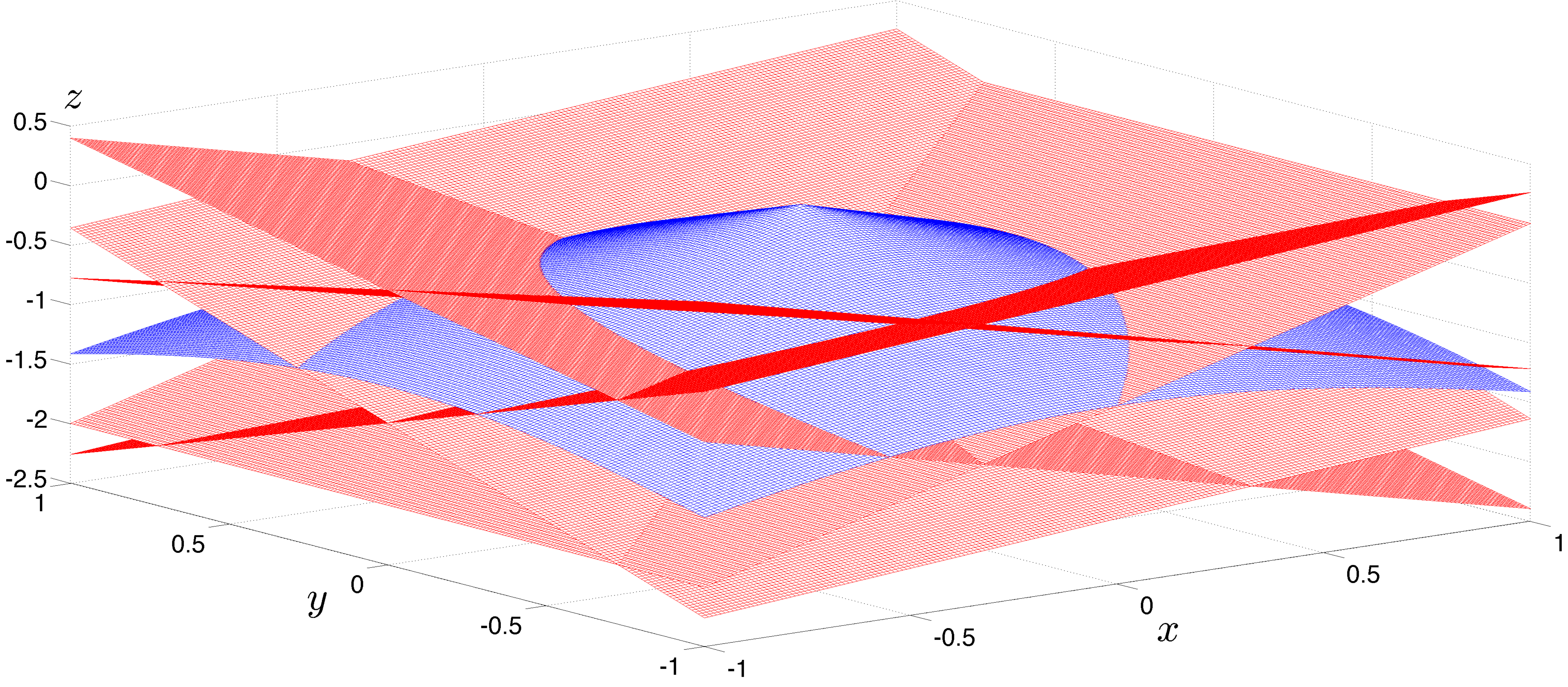}
\end{center}
\caption{The constraints of the problem $\textup{SP}_{\textup{ex}}[\bar\omega]$ with $N=5$ are represented. The blue surface is the set of points such that $z = -\sqrt{x^2 + y^2}$, while the red hyperplanes are the sets of points such that $z = \cos( \bar\delta^{(i)} ) x + \sin( \bar\delta^{(i)} ) y -1$, for $i = 1, 2, ..., 5$. The feasible set is the region above the plotted surfaces and the minimization direction is the vertical one, pointing down.}
\label{fig:SP1}
\end{figure}

We show that all the sampled constraints $z \geq \cos( \bar\delta^{(i)} ) x + \sin( \bar\delta^{(i)} ) y$, for $i = 1, 2, ..., N$, are support constraints, making it impossible to bound Helly's dimension by some $\zeta < N$.

We first compute the optimal value $ J_{ \textup{ex} }^{\star}[\bar\omega]$ of $\textup{SP}_{\textup{ex}}[\bar\omega]$ in \eqref{eq:counterexample}. By symmetry and regularity arguments (i.e. continuity of both the objective function and the constraints in the decision variable), an optimizer $(x_N^\star, y_N^\star, z_N^\star)$ can be computed as the intersection of any two adjacent hyperplanes, say $\left\{(x,y,z) \in \R^3 \mid z = \cos( \bar\delta^{(i)} ) x + \sin( \bar\delta^{(i)} ) y -1 \right\}$ for $i = 1, 2$, and the surface $\left\{ (x,y,z) \in \R^3 \mid z = -\sqrt{ x^2 + y^2 } \right\}$. Since $\bar\delta^{(1)} = 0$ and $\bar\delta^{(2)} = \frac{2 \pi}{N} =: \theta_N \in (0, \pi/2)$, the optimal value and an optimizer can be computed by solving the system of equations:
\begin{equation}\label{eq:system-counterexample}
z = -\sqrt{ x^2 + y^2 } = x-1 = \cos(\theta_N) x + \sin( \theta_N ) y - 1.
\end{equation}
From the second and the third equations of \eqref{eq:system-counterexample}, we get that $y = \frac{\sin(\theta_N)}{ 1+\cos(\theta_N) } x$ and hence from the first equation of \eqref{eq:system-counterexample} we finally get: $\left( \frac{\sin(\theta_N)}{ 1+\cos(\theta_N) } \right)^2 x^2 + 2x -1 = 0$. Therefore an optimizer is
\begin{equation}\label{eq:solution-counterexample}
x_N^\star = \frac{ \sqrt{1+\left( \frac{\sin(\theta_N)}{ 1+\cos(\theta_N) } \right)^2} -1 }{\left( \frac{\sin(\theta_N)}{ 1+\cos(\theta_N) } \right)^2}, \ y_N^\star = \frac{\sin( \theta_N )}{1+\cos(\theta_N)} x_N^\star, \ z_N^\star = x_N^\star - 1
\end{equation}
and the optimal cost is $J_{\textup{ex}}^\star[\bar\omega] = z_N^\star$.

We then remove the sample $\bar\delta^{(2)} = \frac{2 \pi}{N}$, and hence consider the problem $\textup{SP}_{\textup{ex}}[\bar\omega \setminus \bar\delta^{(2)}]$. The optimizer is now unique and lies in the intersection of the hyperplanes \\
$\left\{(x,y,z) \in \R^3 \mid z = \cos( \bar\delta^{(i)} ) x + \sin( \bar\delta^{(i)} ) y -1 \right\}$, for $i = 1, 3$, and the surface\\
$\left\{ (x,y,z) \in \R^3 \mid z = -\sqrt{ x^2 + y^2 } \right\}$. We just need to solve the system of equations \eqref{eq:system-counterexample}, but with $\bar\delta^{(3)} := 2 \theta_N = \frac{4 \pi}{N}$ in place of $\theta_N$ in the third equation. Therefore we obtain almost the same solution in \eqref{eq:system-counterexample}, but with $2 \theta_N$ in place of $\theta_N$. Since the optimal cost 
$$J_{\textup{ex}}^\star[\bar\omega \setminus \bar\delta^{(2)}] = \frac{ \sqrt{1+\left( \frac{\sin(2 \theta_N)}{ 1+\cos(2 \theta_N) } \right)^2} -1 }{\left( \frac{\sin(2 \theta_N)}{ 1+\cos(2 \theta_N) } \right)^2} - 1$$ is strictly smaller than $J_{\textup{ex}}^\star[\bar\omega]$ (as $x_{N+1}^\star < x_N^\star$ for all $N \geq 5)$, it follows that the constraint associated with $\bar\delta^{(2)}$ is a support constraint. Figure \ref{fig:SP2} shows the optimizer of the problem $\textup{SP}_{\textup{ex}}[\bar\omega \setminus \{\bar\delta^{(2)} \}]$.
\begin{figure} 
\begin{center}
\includegraphics[width = 0.8 \columnwidth]{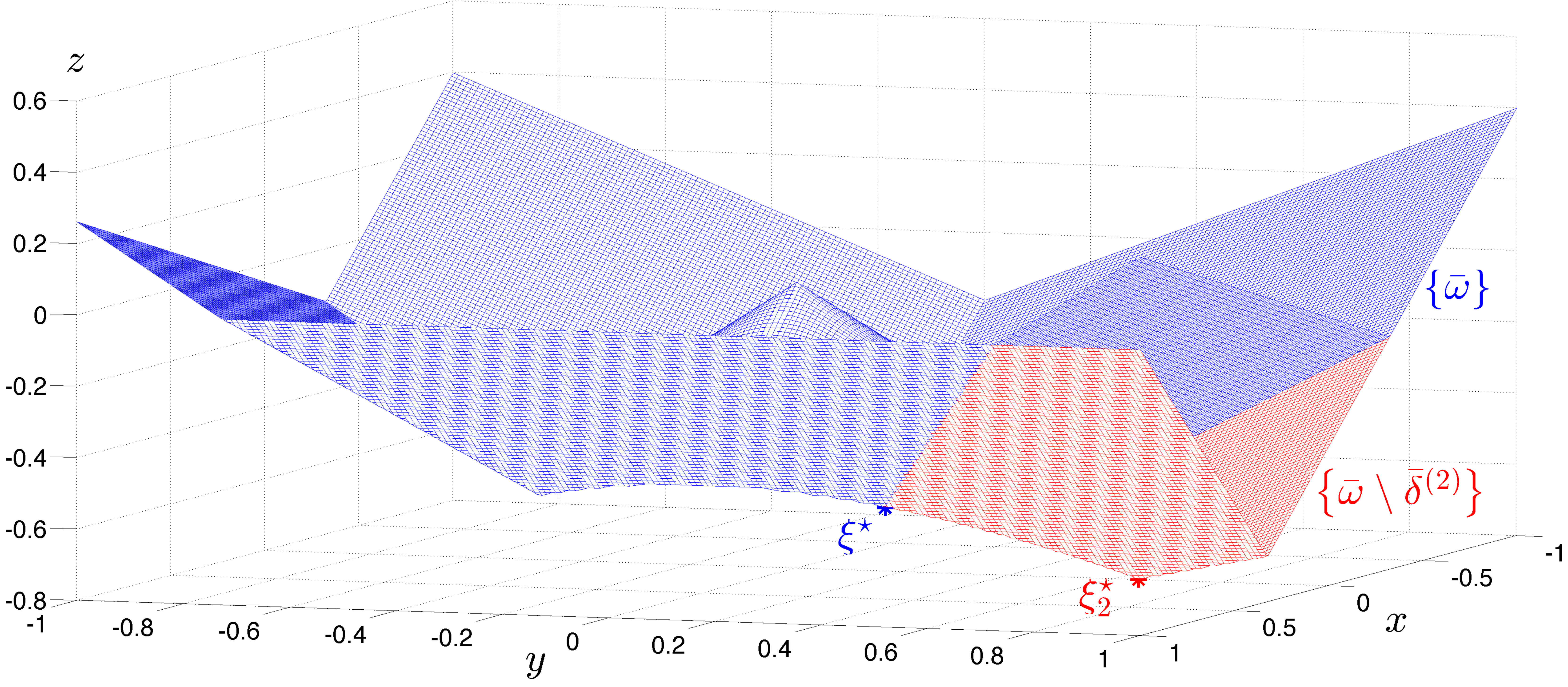}
\end{center}
\caption{The constraints of the problem $\textup{SP}_{\textup{ex}}[\bar\omega \setminus \bar\delta^{(2)}]$ with $N=5$ are represented. The feasible set is the region above the blue surface, which is the set of points such that $z = \max_{ i \in \{1, 3, 4, 5\} }\{ -\sqrt{x^2 + y^2}, \cos( \bar\delta^{(i)} ) x + \sin( \bar\delta^{(i)} ) y -1 \}$. The red dot represents the optimizer $\left( x^\star, y^\star, z^\star\right)$ which has a cost $J\left( x^\star, y^\star, z^\star\right) = J_{\textup{ex}}^\star[\bar\omega \setminus \bar\delta^{(2)}] < J_{\textup{ex}}^\star[\bar\omega]$.}
\label{fig:SP2}
\end{figure}
Because of the symmetry of the problem with respect to rotations around the $z$-axis, we conclude that all the $N$ affine constraints $z \geq \cos( \bar\delta^{(i)} ) x + \sin( \bar\delta^{(i)} ) y -1$, for $i = 1, 2, ..., N$, are support constraints as well, i.e. $J_{\textup{ex}}^\star[\bar\omega \setminus \bar\delta^{(i)}] < J_{\textup{ex}}^\star[\bar\omega]$ for all $i \in \Z[1,N]$. This proves that Helly's dimension cannot be upper bounded by some a-priori fixed $\zeta < N$.

Finally, in view of \cite[Theorem 1]{campi:garatti:08}, it suffices to find at least one probability measure so that the extraction of the above samples $\bar\delta^{(1)}, \bar\delta^{(2)}, \ldots, \bar\delta^{(N)}$ happens with non-zero probability. For instance, this holds true if $\mathbb{P}$ is such that $\mathbb{P}( \{ \bar\delta^{(i)} \} ) = 1/N$ for all $i \in \Z[1,N]$. 
Moreover, it is also possible to have a distribution about the above points $\bar\delta^{(1)}, \bar\delta^{(2)}, \ldots, \bar\delta^{(N)}$ that has a density, but is narrow enough to preserve the property that $J^\star[\bar\omega \setminus \bar\delta^{(2)}] < J^\star[\bar\omega]$.


\section{Proofs} \label{app:proofs}

\subsection*{Proof of Theorem \ref{th:violation-convex-hull}}
Let $\mc{X}_\epsilon := \left\{ x \in \mc{X} \mid \mathbb{P}\left( \{ \delta \in \Delta \mid g(x,\delta)\leq 0 \}\right) \geq 1-\epsilon \right\}$ be the feasibility set of $\textup{CCP}(\epsilon)$ in \eqref{eq:CCP}.
Take any arbitrary $y \in \textup{conv}\left(\mc{X}_{\epsilon}\right)$. It follows from Caratheodory's Theorem \cite[Theorem 17.1]{rockafellar} that there exist $x_1, x_2, ..., x_{n+1} \in \mc{X}_{\epsilon}$ such that $y \in \textup{conv}\left( \left\{ x_1, x_2, ..., x_{n+1} \right\} \right)$, i.e. $y = \sum_{i=1}^{n+1} \alpha_i x_i$ for some $\alpha_1, \alpha_2, ..., \alpha_{n+1} \in [0,1]$ such that $\sum_{i=1}^{n+1} \alpha_i = 1$. 

In the following inequalities, we exploit the convexity of the mapping $x \mapsto g(x,\delta)$ for each fixed $\delta \in \Delta$ from Standing Assumption \ref{ass:standing}.
\begin{equation} \label{eq:inequalities-violation-convex-hull}
\begin{array}{l}
\mathbb{P}\left( \left\{ \delta \in \Delta \mid g(y,\delta)>0 \right\} \right) \ = \ \mathbb{P}\left( \left\{ \delta \in \Delta \mid g( \sum_{i=1}^{n+1} \alpha_i x_i ,\delta)>0 \right\} \right) \\
\leq \ \mathbb{P}\left( \left\{ \delta \in \Delta \mid \sum_{i=1}^{n+1} \alpha_i g(  x_i ,\delta)>0 \right\} \right) \ \leq \ \mathbb{P}\left( \left\{ \delta \in \Delta \mid \max_{i \in \Z{[1,n+1]}} \alpha_i g(  x_i ,\delta)>0 \right\} \right) \ = \\
 \mathbb{P}\left( \bigcup_{i=1}^{n+1} \left\{ \delta \in \Delta \mid g(  x_i ,\delta)>0 \right\} \right) \ \leq \ \sum_{i=1}^{n+1} \mathbb{P}\left( \left\{ \delta \in \Delta \mid g(  x_i ,\delta)>0 \right\} \right) \ \leq \ (n+1)\epsilon.
\end{array}
\end{equation}
The last inequality follows from the fact that $x_1, x_2, ... , x_{n+1} \in \mc{X}_{\epsilon}$. 

Since $y \in \textup{conv}\left(\mc{X}_{\epsilon}\right)$ has been chosen arbitrarily, it follows that $V( \textup{conv}\left(\mc{X}_{\epsilon}\right) ) \leq (n+1) \epsilon$.
\hfill $\blacksquare$

\subsection*{Proof of Lemma \ref{lem:many-xis}}
\noindent $\textstyle \mathbb{P}^N\left( \left\{ \omega \in \Delta^N \mid V\left( \{ x_1^\star(\omega), ..., x_M^\star(\omega) \}\right) > \epsilon  \right\} \right) =
 \textstyle \mathbb{P}^N\left( \bigcup_{j=1}^M \left\{ \omega \in \Delta^N \mid V\left( \{ x_j^\star(\omega)\}\right) > \epsilon  \right\} \right) \leq  \\
\textstyle \sum_{k=1}^M \mathbb{P}^N\left(  \left\{ \omega \in \Delta^N \mid V\left( \{ x_k^\star(\omega)\}\right) > \epsilon  \right\} \right) \leq \sum_{k=1}^M \beta_k
$, where the last inequality follows from Assumption \ref{ass:many-xis}.
\hfill $\blacksquare$


\subsection*{Proof of Theorem \ref{th:bar-Phi}}

For all $\omega \in \Delta^N$, from the definition of the supremum $V(\X) = \sup_{ x \in \X_M(\omega) } V(\{ x \})$ it holds that for all $\epsilon' >0$ there exists $\xi_M^\star(\omega) \in \X_M(\omega) =  \textup{conv}\left(\{x_1^\star(\omega), x_2^\star(\omega), ..., x_M^\star(\omega) \}\right)$ such that
\begin{equation}\label{eq:sup-V}
V(\X_M(\omega)  ) = \sup_{ x \in \X_M(\omega) } V(\{ x \}) < V( \xi_M^\star(\omega) ) + \epsilon' .
\end{equation}

Now, for all $\omega \in \Delta^N$, we denote by $\mc{I}(\omega) \subset \Z{[1,M]} $ the set of indices of cardinality $|\mc{I}(\omega)| = \min\{n+1, M\}$, with ``minimum lexicographic order''\footnote{With ``minimum lexicographic order'' we mean the following ordering: $ \{i_1, i_2, ..., i_n\} < \{ j_1, j_2, ..., j_n \}$ if there exists $k \in \Z{[1,n]}$ such that $i_1 = j_1$, ..., $i_{k-1} = j_{k-1}$, and $i_k < j_k$.}, such that we have the inclusion $\xi_M^\star(\omega) \in \textup{conv}\left(\{x_j^\star(\omega) \mid j \in \mathcal{I}(\omega) \}\right)$. Since ${\X}_M(\omega)$ is convex and compact, it follows from Caratheodory's Theorem \cite[Theorem 17.1]{rockafellar} that such a set of indices $\mathcal{I}(\omega)$ always exists. It also follows that there exists a unique set of coefficients $\alpha_1(\omega), \alpha_2(\omega), ..., \alpha_{n+1}(\omega) \in [0,1]$ such that \\
$\sum_{ j \in \mc{I}(\omega) } \alpha_j(\omega) = 1 $ and
\begin{equation} \label{eq:xi-star}
\xi_M^\star(\omega) = \sum_{j \in \mc{I}(\omega)} \alpha_j(\omega) x_j^\star(\omega).
\end{equation}

In the following inequalities, we exploit \eqref{eq:sup-V}, \eqref{eq:xi-star} and the convexity of the mapping $x \mapsto g(x,\delta)$ for each fixed $\delta \in \Delta$ from Standing Assumption \ref{ass:standing}, and we can take $\epsilon' \in (0, \epsilon)$ without loss of generality.
\begin{equation} \label{eq:inequalities-to-barPhi}
\begin{array}{l}
\displaystyle \mathbb{P}^N\left( \left\{ \omega \in \Delta^N \mid V\left( {\X}_M(\omega) \right) > \epsilon  \right\} \right) 
\displaystyle = \ \textstyle \mathbb{P}^N\left( \left\{ \omega \in \Delta^N \mid \sup_{ x \in \X_M(\omega) } V( \{x\} )  > \epsilon  \right\} \right)  \\
\displaystyle \leq \ \mathbb{P}^N\left( \left\{ \omega \in \Delta^N \mid {V}\left( \{\xi_M^\star(\omega)\} \right) > \epsilon - \epsilon' \right\}  \right) \\
\displaystyle =  \ \textstyle \mathbb{P}^N\left( \left\{ \omega \in \Delta^N \mid \mathbb{P}\left( \left\{ \delta \in \Delta \mid g\left( \sum_{j \in \mc{I}(\omega)} \alpha_j(\omega) x_j^\star(\omega), \delta \right) > 0 \right\} \right)      > \epsilon - \epsilon'  \right\} \right)  \\
\displaystyle \leq \ \textstyle \mathbb{P}^N\left( \left\{ \omega \in \Delta^N \mid \mathbb{P}\left( \left\{ \delta \in \Delta \mid \sum_{j \in \mc{I}(\omega)} \alpha_j(\omega) g\left( x_j^\star(\omega), \delta \right) > 0 \right\} \right)      > \epsilon - \epsilon'  \right\} \right) \\
\displaystyle  \leq \ \textstyle \mathbb{P}^N\left( \left\{ \omega \in \Delta^N \mid \mathbb{P}\left( \left\{ \delta \in \Delta \mid \max_{j \in \mc{I}(\omega)} g\left( x_j^\star(\omega), \delta \right) > 0 \right\} \right)      > \epsilon - \epsilon' \right\} \right) \\
\displaystyle  = \ \textstyle \mathbb{P}^N\left( \left\{ \omega \in \Delta^N \mid \mathbb{P}\left( \bigcup_{j \in \mc{I}(\omega)} \left\{ \delta \in \Delta \mid g\left( x_j^\star(\omega), \delta \right) > 0 \right\} \right)      > \epsilon - \epsilon'  \right\} \right) \\
\displaystyle  \leq \ \textstyle \mathbb{P}^N\left( \left\{ \omega \in \Delta^N \mid \sum_{j \in \mc{I}(\omega)} \mathbb{P}\left( \left\{ \delta \in \Delta \mid g\left( x_j^\star(\omega), \delta \right) > 0 \right\} \right)      > \epsilon - \epsilon'  \right\} \right) \\
\displaystyle  \leq \ \textstyle \mathbb{P}^N\left( \left\{ \omega \in \Delta^N \mid \max_{j \in \mc{I}(\omega)} \mathbb{P}\left( \left\{ \delta \in \Delta \mid g\left( x_j^\star(\omega), \delta \right) > 0 \right\} \right)      >   \frac{\epsilon - \epsilon'}{n+1}  \right\} \right) \\
\displaystyle  = \ \textstyle \mathbb{P}^N\left( \left\{ \omega \in \Delta^N \mid {V}\left( \{ x_j^\star(\omega) \mid j \in \mc{I}(\omega) \} \right)      >   \frac{\epsilon - \epsilon'}{n+1}  \right\} \right) \\
  \leq \ \textstyle \mathbb{P}^N\left( \left\{ \omega \in \Delta^N \mid {V}\left( \{ x_1^\star(\omega), x_2^\star(\omega), ..., x_M^\star(\omega) \} \right)      >   \frac{\epsilon - \epsilon'}{n+1}  \right\} \right).
\end{array}
\end{equation}

Since for all $k \in \Z{[1,M]}$, $x^\star_k(\cdot)$ is the optimizer mapping of $\textsc{SP}_k[\cdot]$ in \eqref{eq:SPk}, from \cite[Theorem 1]{campi:garatti:08}, \cite[Theorem 3.3]{calafiore:10} we have that 
$
\mathbb{P}^N\left( \{ \omega \in \Delta^N \mid V( \{ x_k^\star(\omega) \} ) > \epsilon \} \right) \leq \Phi( \epsilon, n, N ).
$
We now use Lemma \ref{lem:many-xis} with $ \beta_k := \Phi\left( \frac{\epsilon - \epsilon'}{n+1}, n, N \right)$ for all $k \in \Z{[1,M]}$, so that, 
for all $\epsilon' > 0$, we get
\begin{multline*}
\displaystyle \mathbb{P}^N\left( \left\{ \omega \in \Delta^N \mid V\left( {\X}_M(\omega) \right) > \epsilon  \right\} \right) \leq \\
\textstyle \mathbb{P}^N\left( \left\{ \omega \in \Delta^N \mid {V}\left( \{ x_1^\star(\omega), x_2^\star(\omega), ..., x_M^\star(\omega) \} \right)      >   \frac{\epsilon - \epsilon'}{n+1}  \right\} \right) \leq  M \Phi\left( \frac{\epsilon - \epsilon'}{n+1}, n, N \right)
\end{multline*}
Then, since for all $n, N \geq 1$ the mapping $\epsilon \mapsto \Phi(\epsilon, n, N)$ is continuous, we have that \\
$\limsup_{\epsilon' \rightarrow 0} M \Phi\left( \frac{\epsilon - \epsilon'}{n+1}, n, N \right) = \lim_{\epsilon' \rightarrow 0} M \Phi\left( \frac{\epsilon - \epsilon'}{n+1}, n, N \right) = M \Phi\left( \frac{\epsilon}{n+1}, n, N \right)$, which proves \eqref{eq:bar-Phi-M}.
\hfill $\blacksquare$

\subsection*{Proof of Corollary \ref{cor:bar-Phi-best}}
It follows from Carath\'eodory's Theorem \cite[Theorem 17.1]{rockafellar} that, for each $\omega \in \Delta^N$,  there exist the sets $\X_M^{(i)}(\omega) := \textup{conv}\left( \left\{ x_k^\star(\omega) \mid k \in \mc{I}_i \right\} \right)$, for $i = 1, 2, ... , {M \choose n+1}$, where each $\mc{I}_i$ is a set of indices of cardinality $n+1$, such that $\X_M(\omega) = \bigcup_{i=1}^{ {M \choose n+1} } \X_M^{(i)}(\omega)$. Therefore we can write
\begin{equation*}
\begin{array}{l}
\textstyle \mathbb{P}^N\left( \left\{ \omega \in \Delta^N \mid \sup_{x \in \X_M(\omega)} V( \{ x\} ) > \epsilon \right\} \right)  \\
=  \ \textstyle \mathbb{P}^N\left( \left\{ \omega \in \Delta^N \mid \max_{i \in \Z\left[1, {M \choose n+1}\right]} \ \sup_{x \in \X_M^{(i)}(\omega)} V( \{ x\} ) > \epsilon  \right\} \right) \\
=  \ \textstyle \mathbb{P}^N\left( \bigcup_{i=1}^{{M \choose n+1}} \left\{ \omega \in \Delta^N \mid \sup_{x \in \X_M^{(i)}(\omega)} V( \{ x\} ) > \epsilon  \right\} \right) \\
\leq  \ \textstyle \sum_{i=1}^{{M \choose n+1}} \mathbb{P}^N\left( \left\{ \omega \in \Delta^N \mid \sup_{x \in \X_M^{(i)}(\omega)} V( \{ x\} ) > \epsilon  \right\} \right) \\
\leq \  \textstyle  {M \choose n+1} \mathbb{P}^N\left( \left\{ \omega \in \Delta^N \mid \sup_{x \in \X_M^{(1)}(\omega)} V( \{ x\} ) > \epsilon  \right\} \right),
\end{array}
\end{equation*}
where in the last inequality we consider the first set of indices without loss of generality, similarly to \cite[Proof of Theorem 3.3, pag. 3435]{calafiore:10}.
It follows from Theorem \ref{th:bar-Phi} and \cite[Equation (24)]{calafiore:lyons:13} that for all $\epsilon' >0$ there exists $\xi^\star(\omega) \in \X_M^{(1)}(\omega) = \textup{conv}\left( \left\{ x_1^\star(\omega), x_2^\star(\omega), ...,  x_{n+1}^\star(\omega) \right\}\right)$ such that 
\begin{multline*}
\mathbb{P}^N\left( \left\{ \omega \in \Delta^N \mid V( \X_M(\omega) ) > \epsilon \right\} \right) \leq  
\\ {M \choose n+1} \mathbb{P}^N\left( \left\{ \omega \in \Delta^N \mid V( \{ \xi^\star(\omega) \} ) > \epsilon - \epsilon'  \right\} \right) \leq {M \choose n+1} \Phi\left( \epsilon - \epsilon', \zeta (n+1), N \right)
\end{multline*}
and hence, after taking the $\limsup_{\epsilon' \rightarrow 0}$ on both sides of the inequality, we finally get the inequality
$$ \mathbb{P}^N\left( \left\{ \omega \in \Delta^N \mid V( \X_M(\omega) ) > \epsilon \right\} \right) \leq {M \choose n+1} \Phi\left( \epsilon, \zeta(n+1), N \right).$$
\hfill $\blacksquare$

\subsection*{Proof of Corollary \ref{cor:bar-Phi}}
If follows from \eqref{eq:bar-Phi-M} that we need to find $N$ such that $\Phi\left( \frac{\epsilon}{n+1}, n, N \right) < \beta/M$.
The proof follows similarly to \cite[Proof of Theorem 3]{alamo:tempo:luque:ramirez:13}.
\hfill $\blacksquare$

\subsection*{Proof of Theorem \ref{th:bar-Phi-1}}
The proof is similar to the proof of Theorem \ref{th:bar-Phi}.
\hfill $\blacksquare$

\subsection*{Proof of Corollary \ref{cor:bar-Phi-1}}
The proof is similar to the proof of Corollary \ref{cor:bar-Phi}.
\hfill $\blacksquare$


\subsection*{Proof of Theorem \ref{th:N-MPC}}
For each $j \in \Z{[1,M]}$, we consider the random convex problem $\textup{SP}_j^1[\cdot]$ in \eqref{eq:SPk-1}, with unique optimizer mapping $\lambda_j^\star(\cdot)$. Since the dimension of the decision variable is $1$, i.e. $u_j^\star(\cdot) := \hat{u}_0 + \lambda_j^\star(\cdot) c_j$, it follows from \cite[Theorem 1]{campi:garatti:08}, \cite[Theorem 3.3]{calafiore:10} that, for all $j \in \Z{[1,M]}$, we have 
$$ \mathbb{P}^N\left( \left\{ \omega \in \mc{V}^N \mid  V^{ \textup{MPC} }( \{ u_j^\star(\omega) \}) > \epsilon  \right\} \right) \leq \Phi( \epsilon, 1, N ).$$

Then, from Lemma \ref{lem:many-xis} we have that: \\
$ \mathbb{P}^N\left( \left\{ \omega \in \mc{V}^N \mid  V^{ \textup{MPC} }( \{ u_1^\star(\omega), u_2^\star(\omega), ..., u_M^\star(\omega) \}) > \epsilon  \right\} \right) \leq M \Phi\left( \epsilon, 1, N \right)$. 

We now notice that the CCP in \eqref{eq:SMPC} is of the same form of \eqref{eq:CCP-general}, with the constraints \\
$\mathbb{P}^{k} \left( \left\{ \mathbf{v} \in \mathcal{V}^k \mid \phi\left(k; x, \mathbf{\cdot}, \mathbf{v}  \right) \notin \X  \right\} \right) \leq \epsilon$, for $k \geq 2$, in place of $h(\cdot) \leq 0$. Therefore to conclude the proof we just have to follow the steps of Remark \ref{rem:general-CCP} and the proof of Theorem \ref{th:bar-Phi} with $\{u_1^\star(\omega), \ldots u_M^\star(\omega)\}$ in place of $\{x_1^\star(\omega), \ldots x_M^\star(\omega)\}$, and finally derive the sample size $N$ according to \eqref{eq:admissible-K}.
\hfill $\blacksquare$

\subsection*{Proof of Corollary \ref{cor:expected-violations}}
Since the sample size $N$ satisfies \eqref{eq:admissible-K}, the proof follows from \cite[Section 4.2]{schildbach:fagiano:frei:morari:13}.
\hfill $\blacksquare$

\section{Measurability of optimal value and of optimal solutions} \label{app:measurability}
In this section, we adopt the following notion of measurability from \cite[Section 2]{grammatico:subbaraman:teel:13}. Let $\mathcal{B}(\R^n)$ denote the Borel field, the subsets of $\R^n$ generated from open subsets of $\R^n$ through complements and finite countable unions. A set $F \subset \R^n$ is measurable if $F \in \mathcal{B}(\R^n)$. A set-valued mapping $M: \mathbb{R}^n \rightrightarrows \R^m$ is measurable \cite[Definition 14.1]{rockafellar:wets} if for each open set $\mathcal{O} \subset \R^m$ the set $M^{-1}( \mathcal{O} ) := \{ v \in \R^n \mid M(v) \cap \mathcal{O} \neq \varnothing \}$ is measurable. When the values of $M$ are closed, measurability is equivalent to $M^{-1}( \mathcal{C} )$ being measurable for each closed set $\mathcal{C} \in \R^m$ \cite[Theorem 14.3]{rockafellar:wets}. 
Let $\left( \Omega, \mathcal{F}, \mathbb{P}\right)$ be a probability space, where $\mathbb{P}$ is a probability measure on $\R^n$. A set $F \subset \R^n$ is universally measurable if it belongs to the Lebesgue completion of $\mathcal{B}(\R^n)$. A set-valued mapping $M: \mathbb{R}^n \rightrightarrows \R^m$ is universally measurable if the set $M^{-1}( \mathcal{S} )$ is universally measurable for all $\mathcal{S} \in \mathcal{B}(\R^m)$ \cite[Section 7.1, pag. 68]{bogachev2}. If $\varphi: \R^m \rightarrow \R \cup \{ \pm \infty \}$ is a (universally) measurable function, then the integral $I[\varphi] := \int_{ \R^m } \varphi( \omega) \mathbb{P}( d \omega ) $ is (nearly) well defined \cite[Chapter 14, pag. 643]{rockafellar:wets}.



The following result shows (near) well definiteness of the stated probability integrals.
\begin{theorem} \label{th:integrals-well-defined}
For all $x \in \mc{X}$, the probability integral $\mathbb{P}\left( \left\{ \delta \in \Delta \mid g(x,\delta)\leq 0 \right\} \right) $ is well defined.
For any measurable set-valued mapping $\X: \Delta^N \rightrightarrows \R^n$ and $\epsilon \in (0,1)$, the probability integral $\mathbb{P}^N\left( \left\{ \omega \in \Delta^N \mid V( \X(\omega) ) > \epsilon \right\} \right)$ is nearly well defined.
\qed
\end{theorem}

\begin{proof}
From Standing Assumption \ref{ass:standing}, we have that $g$ is a lower semicontinuous convex integrand, and hence a normal integrand \cite[Proposition 14.39]{rockafellar:wets}. Therefore, for all $x \in \mc{X}$, the set $\{ \delta \in \Delta \mid g(x,\delta) \leq 0 \}$ is measurable \cite[Proposition 14.33]{rockafellar:wets} and in turn the probabilistic measure $\mathbb{P}\left( \left\{ \delta \in \Delta \mid g(x,\delta)\leq 0 \right\} \right) $ is well defined.

For the second statement, we show that, for all set-valued measurable mappings $\X$, the mapping $\textstyle \omega \mapsto \sup_{ x \in \X(\omega)} V( \{ x \} )$ is universally measurable.
Since $g$ is a normal integrand, for any finite non-negative measure $\mu$ on $\mc{X} \subseteq \R^n$, we have that $g$ is jointly $\left(\mathbb{P} \otimes \mu\right)$-measurable \cite[Corollary 14.34]{rockafellar:wets}. Indeed, the set $\mc{A} := \{ (x,\delta) \in \mc{X} \times \Delta \mid g(x,\delta) \leq 0 \}$ is $\left(\mathbb{P} \otimes \mu\right)$-measurable \cite[Proof of Corollary 14.34]{rockafellar:wets}, and in turn the mapping $(x,\delta) \mapsto \mathds{1}_{\mc{A}}(x,\delta)$ is measurable.
It then follows from Fubini's Theorem \cite[Theorem 8.8 (a)]{rudin} that the mapping $x \mapsto \int_{\Delta} \mathds{1}_{\mc{A}}(x,\delta) \mathbb{P}(d \delta) = 	\mathbb{P}\left( \{ \delta \in \Delta \mid g(x,\delta) \leq 0 \}\right)$ is measurable, and in turn $x \mapsto V(\{x \}) = 1 - \mathbb{P}\left( \{ \delta \in \Delta \mid g(x,\delta) \leq 0 \}\right)$ is measurable as well \cite[Proposition 14.11 (c)]{rockafellar:wets}.
Since $V$ is measurable, it follows from \cite[Theorem 2.17 (a)]{stinchcombe:white:92} that $\omega \mapsto \sup_{ x \in \X(\omega) } V( \{ x\} ) $ is analytic and hence universally measurable \cite[Fact 2.9]{stinchcombe:white:92}.
\end{proof}

\begin{remark} \label{rem:near-measurability}
According to the proof of Theorem \ref{th:integrals-well-defined}, the mapping $\omega \mapsto \bar{V}( \X(\omega))$ is not measurable, but only nearly measurable. However, near measurability is sufficient for the purposes of most applications, for instance in game-theory and econometrics, see \cite{stinchcombe:white:92} and the references therein.


Notice however that the upper closure of $V$, i.e. $\bar{V}(\{x\}) := \lim\sup_{ y \rightarrow x } V( \{y\} )$, is such that the integral $\mathbb{P}^N\left( \left\{ \omega \in \Delta^N \mid \bar{V}( \X(\omega) ) > \epsilon \right\} \right)$ is well defined. In fact, since $\bar{V}$ is upper semicontinuous by construction, $-\bar{V}$ is an autonomous, lower semicontinuous, normal integrand \cite[Example 14.30]{rockafellar:wets}. Then it follows from \cite[Example 14.32, Theorem 14.37]{rockafellar:wets} that the mapping $\omega \mapsto \bar{V}( \X(\omega)) := \sup_{ x \in \X(\omega) } \bar{V}(\{x\})$ is measurable.
\hfill $\qed$
\end{remark}

We can now show the following result on the measurability of optimal value and of optimal solutions of ${\textsc{SP}}[\omega]$ in \eqref{eq:SP}, which means that they actually are random variables.
\begin{theorem}\label{th:measurability}
Let $J^\star: \Delta^N \rightarrow \R$ and $\mc{X}^\star: \Delta^N \rightrightarrows \mc{X}$ be the mappings such that, for all $\omega \in \Delta^N$,
$J^\star(\omega)$ and $\mc{X}^\star(\omega)$ are, respectively, the optimal value and the set of optimizers of $\textsc{SP}[\omega]$ in \eqref{eq:SP}. Then $J^\star$ is measurable, and $\mc{X}^\star$ is closed-valued and measurable. Moreover, $\mc{X}^\star$ admits a measurable selection, i.e., there exists a measurable mapping $x^\star: \Delta^N \rightarrow \mc{X}$ such that $x^\star(\omega) \in \mc{X}^\star(\omega)$ for all $\omega \in \Delta^N$.
\qed
\end{theorem}

\begin{proof}
Since the mapping $x \mapsto g(x,\delta)$ is convex and lower semicontinuous for each $\delta$, and the mapping $\delta \mapsto g(x,\delta)$ is measurable for each $x$, we have that $g$ is a lower semicontinuous integrand and hence a normal integrand \cite[Definition 14.27, Proposition 14.39]{rockafellar:wets}.
For all $i \in \Z[1,N]$, we consider the lower semicontinuous convex, and hence normal \cite[Proposition 14.39]{rockafellar:wets}, integrand $g_i: \mc{X} \times \Delta^N \rightarrow \R$ defined as $g_i(x,\omega) = g_i\left(x,(\delta_1, \delta_2, ..., \delta_N) \right) := g(x,\delta_i)$. 
Then we consider the mapping $\bar{g}: \mc{X} \times \Delta^N \rightarrow \R$ defined as $\bar{g}(x,\omega) := \max_{i \in \Z[1,N]} g_i(x,\omega)$, which is a normal integrand because the pointwise maximum of the normal integrands $g_1$, $g_2$, ..., $g_N$ \cite[Proposition 14.44 (a)]{rockafellar:wets}.
We now consider the set-valued mapping $\mc{C}: \Delta^n \rightrightarrows \mc{X}$ defined as $\mc{C}(\omega) := \left\{ x \in \mc{X} \mid \bar{g}(x,\omega) \leq 0 \right\}$. Since $\bar{g}$ is a normal integrand, it follows from \cite[Proposition 14.33]{rockafellar:wets} that the level-set mapping $\mc{C}$ is closed-valued and measurable.
Thus, we can define the indicator integrand $\mathds{1}_{\mc{C}}: \mc{X} \times \Delta^N \rightarrow \{0, \infty\}$ as $\mathds{1}_{\mc{C}}(x,\omega) = \mathds{1}_{ \mc{C}(\omega) }(x) := \{ 0 \textup{ if }  x \in \mc{C}(\omega), \ \infty \textup{ otherwise} \}.$
Since $\mc{C}$ is closed-valued and measurable, the mapping $\mathds{1}_{\mc{C}}$ is a normal integrand \cite[Example 14.32]{rockafellar:wets}.
Now, the problem $\textup{SP}[\omega]$ in \eqref{eq:SP} can be written as $\min_{x \in \mc{X}} c^\top x$ sub. to $x \in \mc{C}(\omega)$, which is equivalent \cite[Section 1.A]{rockafellar:wets} to $\min_{ x \in \R^n } J(x) + \mathds{1}_{\mc{C}}(x,\omega)$. We notice that the mapping $(x,\omega) \mapsto \varphi(x,\omega) := J(x) + \mathds{1}_{\mc{C}}(x,\omega)$ is a normal integrand as $J$ is lower semicontinuous \cite[Example 14.30, Example 14.32, Proposition 14.44 (c)]{rockafellar:wets}.
It finally follows from \cite[Theorem 14.37]{rockafellar:wets} that the optimal value mapping $\omega \mapsto J^\star(\omega) := \inf_{x \in \R^n} \varphi(x,\omega) $ is measurable; also, the set-valued mapping 
$\omega \mapsto \mc{X}^\star(\omega) := \arg\min_{x \in \R^n} \varphi(x,\omega) $ is closed-valued and measurable. Moreover, the set $\left\{\omega \in \Delta^N \mid \mc{X}^\star(\omega) \neq \varnothing \right\}$ is measurable, and it is possible for each $\omega \in \Delta^N$ to select a minimizing point $x^\star(\omega)$ in such a manner that the mapping $\omega \mapsto x^\star(\omega)$ is measurable \cite[Corollary 14.6, Theorem 14.37]{rockafellar:wets}.
\end{proof}

In the following result, we show that if the set of optimizers $\mc{X}^\star$ of $\textup{SP}$ in \eqref{eq:SP} is not a singleton, convex and lower semicontinuous tie-break rules $\varphi$ are sufficient to guarantee measurability of the optimizer $x^\star$ (whenever it is unique). Applying a tie-break rule $\varphi$ basically means to solve the following program, where $J^\star(\omega)$ is the optimal value of $\textup{SP}[\omega]$ in \eqref{eq:SP}.
\begin{equation}\label{eq:SP-tie}
{\textsc{SP}}_{\textup{t-b}}[\omega]: \ \left\{
\begin{array}{l}
\displaystyle \min_{ x \in \mc{X} } \ \varphi(x) \\
\begin{array}{ll}
\textup{sub. to:} & g\left( x, \delta^i \right) \leq 0 \ \ \forall i \in \Z[1,N]\\
 & J(x) \leq J^\star(\omega) \\
\end{array}
\end{array}
\right.
\end{equation}

\begin{corollary}
Let $\varphi: \R^n \rightarrow \R$ be a convex and lower semicontinuous function. Let $J^\star: \Delta^N \rightarrow \R$ and $x_{\textup{t-b}}^\star: \Delta^N \rightarrow \mc{X}$ be such that, for all $\omega \in \Delta^N$, $J^\star(\omega)$ and $x_{\textup{t-b}}^\star(\omega)$ are, respectively, the optimal value of $\textup{SP}[\omega]$ in \eqref{eq:SP} and the unique optimal solution of ${\textsc{SP}}_{\textup{t-b}}[\omega]$ in \eqref{eq:SP-tie}. Then $x_{\textup{t-b}}^\star$ is measurable.
\qed
\end{corollary}

\begin{proof}
We first define the normal integrand \\
$\bar{g}(x,\omega) = \bar{g}(x,(\delta_1, \delta_2, ..., \delta_N)) := \max_{ i \in \Z[1,N] } g_i(x,\omega) = \max_{ i \in \Z[1,N] } g( x, \delta_i )$, as in the proof of Theorem \ref{th:measurability}. Since $J$ is lower semicontinuous, it is an autonomous integrand and hence a normal integrand \cite[Example 14.30]{rockafellar:wets}; moreover, since $J^\star$ is measurable from Theorem \ref{th:measurability}, it is a (Carath\'eodory) normal integrand as well \cite[Example 14.29]{rockafellar:wets}. Therefore also the mapping $(x,\omega) \mapsto J(x) - J^\star(\omega)$ is a normal integrand \cite[Proposition 14.44 (c)]{rockafellar:wets}, and in turn, the mapping $\bar{\bar{g}}(x,\omega) := \max\{ \bar{g}(x,\omega), J(x) - J^\star(\omega) \}$ is a normal integrand as well.
Then, we can just follow the proof of Theorem \ref{th:measurability} with $\bar{\bar{g}}$ in place of $\bar{g}$.
\end{proof}

\begin{remark} \label{rem:unique-measurable}
In \eqref{eq:SP-tie}, if $J$ is convex and $\varphi$ is strictly convex then an optimal solution $x_{\textup{t-b}}^\star(\omega)$ of ${\textsc{SP}}_{\textup{t-b}}[\omega]$ is the unique optimal solution.
\qed
\end{remark}

We finally mention that the convex hull of measurable singleton mappings is measurable as well, so that $\mathbb{P}^N \left( \left\{\omega \in \Delta^N \mid V( \X_M(\omega) ) > \epsilon \right\} \right) $ is well defined from Theorem \ref{th:integrals-well-defined}.
\begin{corollary}
The set-valued mapping $\X_M$ in \eqref{eq:XM} is measurable.
\qed
\end{corollary}
\begin{proof}
According to Theorem \ref{th:measurability} and Remark \ref{rem:unique-measurable}, the unique optimal solutions $x_1^\star$, $x_2^\star$, ..., $x_M^\star$, respectively of $\textup{SP}_1$, $\textup{SP}_2$, ..., $\textup{SP}_M$, are measurable mappings.
Then the proof directly follows as $\X_M$ is the convex-hull set-valued mapping of a countable union of measurable mappings \cite[Proposition 114.11 (b), Example 14.12 (a)]{rockafellar:wets}.
\end{proof}

\bibliographystyle{IEEEtran}
\bibliography{library}

\begin{thebibliography}{10}
\providecommand{\url}[1]{#1}
\csname url@samestyle\endcsname
\providecommand{\newblock}{\relax}
\providecommand{\bibinfo}[2]{#2}
\providecommand{\BIBentrySTDinterwordspacing}{\spaceskip=0pt\relax}
\providecommand{\BIBentryALTinterwordstretchfactor}{4}
\providecommand{\BIBentryALTinterwordspacing}{\spaceskip=\fontdimen2\font plus
\BIBentryALTinterwordstretchfactor\fontdimen3\font minus
  \fontdimen4\font\relax}
\providecommand{\BIBforeignlanguage}[2]{{%
\expandafter\ifx\csname l@#1\endcsname\relax
\typeout{** WARNING: IEEEtran.bst: No hyphenation pattern has been}%
\typeout{** loaded for the language `#1'. Using the pattern for}%
\typeout{** the default language instead.}%
\else
\language=\csname l@#1\endcsname
\fi
#2}}
\providecommand{\BIBdecl}{\relax}
\BIBdecl

\bibitem{apkarian:tuan:00}
P.~Apkarian and H.~D. Tuan, ``{Parameterized LMIÕs in control theory},''
  \emph{SIAM Journal on Control and Optimization}, vol.~38, no.~4, pp.
  1241--1264, 2000.

\bibitem{zhou:doyle:glover:97}
K.~Zhou, J.~Doyle, and F.~Glover, \emph{Robust and optimal control}.\hskip 1em
  plus 0.5em minus 0.4em\relax Prentice Hall, 1997.

\bibitem{bertsekas}
D.~P. Bertsekas, \emph{Dynamic programming and optimal control}.\hskip 1em plus
  0.5em minus 0.4em\relax Athena Scientific, 2005.

\bibitem{beard:saridis:wen:97}
R.~W. Beard, G.~N. Saridis, and J.~T. Wen, ``Galerkin approximations of the
  generalized {Hamilton--Jacobi--Bellman} equation,'' \emph{Automatica},
  vol.~33, no.~12, pp. 2159--2177, 1997.

\bibitem{garcia:prett:morari:89}
C.~Garcia, D.~Prett, and M.~Morari, ``{Model predictive control: theory and
  practice - a survey},'' \emph{Automatica}, vol.~25, pp. 335--348, 1989.

\bibitem{mayne:rawlings:rao:scokaert:00}
D.~Q. Mayne, J.~Rawlings, C.~Rao, and P.~Scokaert, ``Constrained model
  predictive control: stability and optimality,'' \emph{Automatica}, vol.~36,
  pp. 789--814, 2000.

\bibitem{bental:nemirovski:02}
A.~Ben-Tal and A.~Nemirovski, ``On tractable approximations of uncertain linear
  matrix inequalities affected by interval uncertainty,'' \emph{SIAM Journal on
  Optimization}, vol.~12, no.~3, pp. 811--833, 2002.

\bibitem{bertsimas:sim:06}
D.~Bertsimas and M.~Sim, ``Tractable approximations to robust conic
  optimization problems,'' \emph{Mathematical Programming}, vol. 107, pp.
  5--36, 2006.

\bibitem{bental:nemirovski:98}
A.~Ben-Tal and A.~Nemirovski, ``Robust convex optimization,'' \emph{Mathematics
  of Operations Research}, vol.~23, no.~4, pp. 769--805, 1998.

\bibitem{bental:nemirovski:99}
------, ``Robust solutions of uncertain linear programs,'' \emph{Operations
  Research Letters}, vol.~25, no.~1, pp. 1--13, 1999.

\bibitem{prekopa}
Pr\'ekopa, \emph{Stochastic Programming}.\hskip 1em plus 0.5em minus
  0.4em\relax Mathematics and Its Applications. Springer, 1995.

\bibitem{shapiro:dentcheva:ruszcynski}
A.~Shapiro, D.~Dentcheva, and A.~Ruszczy\'nski, \emph{{Lectures on Stochastic
  Programing. Modeling and Theory}}.\hskip 1em plus 0.5em minus 0.4em\relax
  SIAM and Mathematical Programming Society, 2009.

\bibitem{charnes:cooper:symonds:58}
A.~Charnes, W.~W. Cooper, and G.~H. Symonds, ``{Cost horizons and certainty
  equivalents: an approach to stochastic programming of heating oil},''
  \emph{Management Science}, vol.~4, pp. 235--263, 1958.

\bibitem{miller:wagner:65}
L.~B. Miller and H.~Wagner, ``Chance-constrained programming with joint
  constraints,'' \emph{Operations Research}, pp. 930--945, 1965.

\bibitem{nemirovski:shapiro:04}
A.~Nemirovski and A.~Shapiro, ``Scenario approximations of chance
  constraints,'' in \emph{Probabilistic and randomized methods for design under
  uncertainty}.\hskip 1em plus 0.5em minus 0.4em\relax Springer, 2004, pp.
  3--48.

\bibitem{nemirovski:shapiro:06}
------, ``Convex approximations of chance constrained programs,'' \emph{SIAM
  Journal on Optimization}, vol.~17, no.~4, pp. 969--996, 2006.

\bibitem{calafiore:campi:05}
G.~Calafiore and M.~C. Campi, ``Uncertain convex programs: randomized solutions
  and confidence levels,'' \emph{Mathematical Programming}, vol. 102, no.~1,
  pp. 25--46, 2005.

\bibitem{calafiore:campi:06}
------, ``The scenario approach to robust control design,'' \emph{IEEE Trans.
  on Automatic Control}, vol.~51, no.~5, pp. 742--753, 2006.

\bibitem{campi:garatti:08}
M.~C. Campi and S.~Garatti, ``The exact feasibility of randomized solutions of
  robust convex programs,'' \emph{SIAM Journal on Optimization}, vol.~19,
  no.~3, pp. 1211--1230, 2008.

\bibitem{calafiore:10}
G.~C. Calafiore, ``Random convex programs,'' \emph{SIAM Journal on
  Optimization}, vol.~20, no.~6, pp. 3427--3464, 2010.

\bibitem{schildbach:fagiano:morari:13}
G.~Schildbach, L.~Fagiano, and M.~Morari, ``Randomized solutions to convex
  programs with multiple chance constraints,'' \emph{SIAM Journal on
  Optimization (in press). \textup{Available online at:
  \texttt{\url{http://arxiv.org/pdf/1205.2190v2.pdf}}}}, 2013.

\bibitem{campi:garatti:11}
M.~C. Campi and S.~Garatti, ``A sampling-and-discarding approach to
  chance-constrained optimization: feasibility and optimality,'' \emph{Journal
  of Optimization Theory and Applications}, vol. 148, no.~2, pp. 257--280,
  2011.

\bibitem{vapnik:charvonenkis:71}
V.~Vapnik and A.~Chervonenkis, ``On the uniform convergence of relative
  frequencies to their probabilities,'' \emph{Theory of Probability and its
  Applications}, vol.~16, no.~2, pp. 264--280, 1971.

\bibitem{anthony:biggs}
M.~Anthony and N.~Biggs, \emph{Computational Learning Theory}.\hskip 1em plus
  0.5em minus 0.4em\relax Cambridge Tracts in Theoretical Computer Science,
  1992.

\bibitem{vidyasagar:97}
M.~Vidyasagar, \emph{A theory of learning and generalization. With applications
  to neural networks and control systems}.\hskip 1em plus 0.5em minus
  0.4em\relax Springer-Verlag, 1997.

\bibitem{alamo:tempo:camacho:09}
T.~Alamo, R.~Tempo, and E.~F. Camacho, ``Randomized strategies for
  probabilistic solutions of uncertain feasibility and optimazation problems,''
  \emph{IEEE Trans. on Automatic Control}, vol.~54, no.~11, 2009.

\bibitem{calafiore:dabbene:tempo:11}
G.~Calafiore, F.~Dabbene, and R.~Tempo, ``Research on probabilistic methods for
  control system design,'' \emph{Automatica}, vol.~47, pp. 1279--1293, 2011.

\bibitem{tempo:calafiore:dabbene}
R.~Tempo, G.~Calafiore, and F.~Dabbene, \emph{Randomized algorithms for
  analysis and control of uncertain systems}.\hskip 1em plus 0.5em minus
  0.4em\relax Springer-Verlag, 2004.

\bibitem{calafiore:campi:elghaoui:02}
G.~Calafiore, M.~C. Campi, and L.~E. Ghaoui, ``Identification of reliable
  predictor models for unknown systems: a data-consistency approach based on
  learning theory,'' in \emph{IFAC World Congress}, Barcelona, Spain, 2002.

\bibitem{rockafellar:wets}
{R.T. Rockafellar and R.J.B. Wets}, \emph{Variational Analysis}.\hskip 1em plus
  0.5em minus 0.4em\relax Springer, 1998.

\bibitem{alamo:tempo:luque:ramirez:13}
T.~Alamo, R.~Tempo, A.~Luque, and D.~Ramirez, ``The sample complexity of
  randomized methods for analysis and design of uncertain systems,''
  \emph{Automatica (submitted). \textup{Available online at:
  \texttt{\url{http://arxiv.org/pdf/1304.0678v1.pdf}}}}, 2013.

\bibitem{floyd:warmuth:95}
S.~Floyd and M.~Warmuth, ``{Sample compression, learnability, and the
  Vapnik-Charvonenkis dimension},'' \emph{Machine learning}, pp. 1--36, 1995.

\bibitem{levin:69}
V.~L. Levin, ``{Application of E. Helly's theorem to convex programming,
  problems of best approximation and related questions},'' \emph{Math. USSR
  Sbornik}, vol.~8, no.~2, pp. 235--247, 1969.

\bibitem{erdogan:iyengar:06}
E.~Erdogan and G.~Iyengar, ``Ambiguous chance constrained problems and robust
  optimization,'' \emph{Mathematical Programming}, vol. 107, pp. 37--61, 2006.

\bibitem{calafiore:lyons:13}
G.~C. Calafiore and D.~Lyons, ``Random convex programs for distributed
  multi-agent consensus,'' in \emph{IEEE European Control Conference}, 2013.

\bibitem{care:garatti:campi:11}
A.~Car\'e, S.~Garatti, and M.~Campi, ``{FAST: an algorithm for the scenario
  approach with reduced sample complexity},'' in \emph{IFAC World Congress},
  Milano, Italy, 2011, pp. 9236--9241.

\bibitem{campi:calafiore:garatti:09}
M.~Campi, G.~Calafiore, and S.~Garatti, ``Interval predictor models:
  identification and reliability,'' \emph{Automatica}, vol.~45, no.~2, pp.
  382--392, 2009.

\bibitem{campi:garatti:prandini:09}
M.~Campi, S.~Garatti, and M.~Prandini, ``The scenario approach for systems and
  control design,'' \emph{Annual Reviews in Control}, vol.~33, no.~2, pp.
  149--157, 2009.

\bibitem{calafiore:lyons:fagiano:12}
G.~C. Calafiore, D.~Lyons, and L.~Fagiano, ``On mixed-integer random convex
  programs,'' in \emph{Proc. of the IEEE Conf. on Decision and Control}, Maui,
  Hawai'i, USA, 2012, pp. 3508--3513.

\bibitem{mohajerin-esfahani:sutter:lygeros:13}
P.~M. Esfahani, T.~Sutter, and J.~Lygeros, ``Performance bounds for the
  scenario approach and an extension to a class of non-convex programs,''
  \emph{IEEE Trans. on Automatic Control (submitted). \textup{Available online
  at: \texttt{\url{http://arxiv.org/pdf/1307.0345.pdf}}}}, 2013.

\bibitem{vidyasagar:01}
M.~Vidyasagar, ``Randomized algorithms for robust controller synthesis using
  statistical learning theory,'' \emph{Automatica}, vol.~37, pp. 1515--1528,
  2001.

\bibitem{margellos:goulart:lygeros:13}
K.~Margellos, P.~Goulart, and J.~Lygeros, ``On the road between robust
  optimization and the scenario approach for chance constrained optimization
  problems,'' \emph{IEEE Trans. on Automatic Control (accepted).
  \textup{Available online at:
  \texttt{\url{http://control.ee.ethz.ch/index.cgi?page=publications\&action=details\&id=4259}}}},
  2013.

\bibitem{calafiore:fagiano:12}
G.~C. Calafiore and L.~Fagiano, ``{Robust MPC via scenario optimization},''
  \emph{IEEE Trans. on Automatic Control}, vol.~58, no.~1, pp. 219--224, 2012.

\bibitem{schildbach:fagiano:frei:morari:13}
G.~Schildbach, L.~Fagiano, C.~Frei, and M.~Morari, ``The scenario approach for
  stochastic model predictive control with bounds on closed-loop constraint
  violations,'' \emph{Automatica (provisionally accepted). \textup{Available
  online at: \texttt{\url{http://arxiv.org/pdf/1307.5640v1.pdf}}}}, 2013.

\bibitem{grammatico:subbaraman:teel:13}
S.~Grammatico, A.~Subbaraman, and A.~Teel, ``{Discrete-time stochatic
  discrete-time systems: a continuous Lyapunov function implies robustness to
  strictly causal perturbations},'' \emph{Automatica}, vol.~49, pp.
  2939Ð--2952, 2013.

\bibitem{zhang:grammatico:margellos:goulart:lygeros:14ifac}
X.~Zhang, S.~Grammatico, K.~Margellos, P.~Goulart, and J.~Lygeros,
  ``{Randomized nonlinear MPC for uncertain control-affine systems with bounded
  closed-loop constraint violations},'' in \emph{IFAC World Congress
  (submitted). \textup{Available online at:
  \texttt{\url{http://control.ee.ethz.ch/~gsergio/ZhaGraMarGouLyg_IFAC14.pdf}}}},
  Cape Town, South Africa, 2014.

\bibitem{grammatico:zhang:margellos:goulart:lygeros:14acc}
S.~Grammatico, X.~Zhang, K.~Margellos, P.~Goulart, and J.~Lygeros, ``A scenario
  approach to non-convex control design: set-based probabilistic guarantees,''
  in \emph{IEEE American Control Conference (submitted). \textup{Available
  online at:
  \texttt{\url{http://control.ee.ethz.ch/~gsergio/GraZhaMarGouLyg_ACC14.pdf}}}},
  Portland, Oregon, USA, 2014.

\bibitem{chamanbaz:13}
M.~Chamanbaz, F.~Dabbene, R.~Tempo, V.~Venkataramanan, and Q.-G. Wang, ``A
  statistical learning theory approach for uncertain linear and bilinear matrix
  inequalities,'' \emph{Automatica (submitted). \textup{Available online at:
  \texttt{\url{http://arxiv.org/pdf/1305.4952v1.pdf}}}}, 2013.

\bibitem{vrakopoulou:margellos:lygeros:andersson:13}
M.~Vrakopoulou, K.~Margellos, J.~Lygeros, and G.~Andersson, ``{A probabilistic
  framework for reserve scheduling and N-1 security assessment of systems with
  high wind power penetration},'' \emph{IEEE Trans. on Power Systems}, vol.~28,
  no.~4, pp. 3885--3896, 2013.

\bibitem{wang:stengel:05}
Q.~Wang and R.~F. Stengel, ``Robust nonlinear flight control of a high
  performance aircraft,'' \emph{IEEE Trans. on Control Systems Technology},
  vol.~13, pp. 15--26, 2005.

\bibitem{calafiore:dabbene:08}
G.~C. Calafiore and F.~Dabbene, ``Optimization under uncertainty with
  applications to design of truss structures,'' \emph{Structural and
  Multidisciplinary Optimization}, vol.~35, pp. 189--200, 2008.

\bibitem{ishii:basar:tempo:05}
H.~Ishii, T.~Basar, and R.~Tempo, ``Randomized algorithms for synthesis of
  switching rules for multimodal systems,'' \emph{IEEE Transactions on
  Automatic Control}, vol.~50, pp. 754--767, 2005.

\bibitem{alpcan:basar:tempo:05}
T.~Alpcan, T.~Basar, and R.~Tempo, ``Randomized algorithms for stability and
  robustness analysis of high speed communication networks,'' \emph{IEEE
  Transactions on Neural Networks}, vol.~16, pp. 1229Ð--1241, 2005.

\bibitem{ma:sznaier:lagoa:07}
W.~Ma, M.~Sznaier, and C.~M. Lagoa, ``A risk adjusted approach to robust
  simultaneous fault detection and isolation,'' \emph{Automatica}, vol.~43,
  no.~3, pp. 499--504, 2007.

\bibitem{rockafellar}
R.~Rockafellar, \emph{Convex Analysis}.\hskip 1em plus 0.5em minus 0.4em\relax
  Princeton University Press, 1970.

\bibitem{bogachev2}
V.~Bogachev, \emph{Measure theory. Vol. 2}.\hskip 1em plus 0.5em minus
  0.4em\relax Springer, 2000.

\bibitem{rudin}
W.~Rudin, \emph{Real \& complex analysis}.\hskip 1em plus 0.5em minus
  0.4em\relax McGraw-Hill, 1987.

\bibitem{stinchcombe:white:92}
M.~B. Stinchcombe and H.~White, ``Some measurability results for extrema of
  random functions over random sets,'' \emph{Review of Economic Studies},
  vol.~59, pp. 495--512, 1992.

\end{thebibliography}

\end{document}